\documentclass[
amsmath, amssymb, 
aps, pra,
reprint, twocolumn,
superscriptaddress,
longbibliography,
floatfix,
showkeys
]{revtex4-1}

\usepackage{graphicx}
\usepackage{csquotes}

\usepackage{multirow}
\usepackage{amssymb}
\usepackage[urlcolor=blue, colorlinks=true,citecolor=blue]{hyperref}

\usepackage{amsthm}

\usepackage{color}
\usepackage{nicefrac}

\definecolor{orange}{RGB}{255,127,0}

\def\bra#1{\ensuremath{\mathinner{\langle{#1}|}}}
\def\ket#1{\ensuremath{\mathinner{|{#1}\rangle}}}

\newcommand{\braket}[2]{\langle #1|#2\rangle}

\newcommand{\tr}{\text{Tr}}

\newtheorem{theorem}{Theorem}

\newtheorem{definition}{Definition}
\newtheorem{lemma}{Lemma}

\newtheorem{remark}{Remark}

\newtheorem{proposition}{Proposition}
\DeclareMathOperator{\Tr}{Tr}
\DeclareMathOperator{\E}{\mathbb{E}}

\DeclareMathOperator{\Var}{\mathrm{Var}}

\newcommand{\norm}[1]{\left\lVert#1\right\rVert}

\usepackage[noend]{algpseudocode}
\usepackage{algorithm,algorithmicx}

\algrenewcommand\alglinenumber[1]{\sf\scriptsize\color{blue}{#1}}
\algrenewcommand\algorithmicrequire{\textbf{Input:}}
\algrenewcommand\algorithmicensure{\textbf{Output:}}

\input{Qcircuit}

\begin{document}

\title{Power of data in quantum machine learning}

\author{Hsin-Yuan Huang}
	\affiliation{Google Research, 340 Main Street, Venice, CA 90291, USA}
	\affiliation{Institute for Quantum Information and Matter, Caltech, Pasadena, CA, USA}
	\affiliation{Department of Computing and Mathematical Sciences, Caltech, Pasadena, CA, USA}
\author{Michael Broughton}
	\affiliation{Google Research, 340 Main Street, Venice, CA 90291, USA}
\author{Masoud Mohseni}	
	\affiliation{Google Research, 340 Main Street, Venice, CA 90291, USA}
\author{Ryan Babbush}	
	\affiliation{Google Research, 340 Main Street, Venice, CA 90291, USA}
\author{Sergio Boixo}	
	\affiliation{Google Research, 340 Main Street, Venice, CA 90291, USA}
\author{Hartmut Neven}	
	\affiliation{Google Research, 340 Main Street, Venice, CA 90291, USA}
\author{Jarrod R. McClean}
	\email{Corresponding author: jmcclean@google.com}
	\affiliation{Google Research, 340 Main Street, Venice, CA 90291, USA}
\date{\today}     

\begin{abstract}
The use of quantum computing for machine learning is among the most exciting prospective applications of quantum technologies.
However, machine learning tasks where data is provided can be considerably different than commonly studied computational tasks.
In this work, we show that some problems that are classically hard to compute can be easily predicted by classical machines learning from data. 
Using rigorous prediction error bounds as a foundation, we develop a methodology for assessing potential quantum advantage in learning tasks.
The bounds are tight asymptotically and empirically predictive for a wide range of learning models.
These constructions explain numerical results showing that with the help of data, classical machine learning models can be competitive with quantum models even if they are tailored to quantum problems.
We then propose a projected quantum model that provides a simple and rigorous quantum speed-up for a learning problem in the fault-tolerant regime. For near-term implementations, we demonstrate a significant prediction advantage over some classical models on engineered data sets designed to demonstrate a maximal quantum advantage in one of the largest numerical tests for gate-based quantum machine learning to date, up to 30 qubits.
\end{abstract}

\maketitle

\section*{Introduction}

As quantum technologies continue to rapidly advance, it becomes increasingly important to understand which applications can benefit from the power of these devices.  At the same time, machine learning on classical computers has made great strides, revolutionizing applications in image recognition, text translation, and even physics applications, with more computational power leading to ever increasing performance~\cite{halevy2009unreasonable}.  As such, if quantum computers could accelerate machine learning, the potential for impact is enormous.

\begin{figure*}[t!]
\centering
\includegraphics[width=0.83\textwidth]{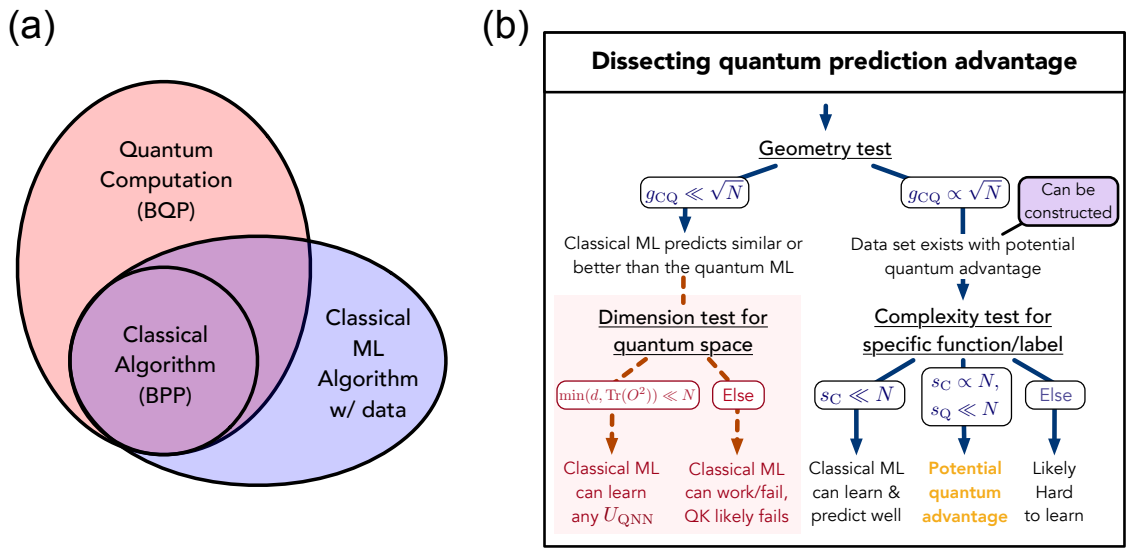}
    \caption{Illustration of the relation between complexity classes and a flowchart for understanding and pre-screening potential quantum advantage. (a) We cartoon the separation between problem complexities that are created by the addition of data to a problem.  Classical algorithms that can learn from data define a complexity class that can solve problems beyond classical computation (BPP), but it is still expected that quantum computation can efficiently solve problems that classical ML algorithm with data cannot. Rigorous definition and proof for the separation between classical algorithms that can learn from data and BPP / BQP is given in Appendix~\ref{sec:powerdatarigor}.
    (b) The flowchart we develop for understanding the potential for quantum prediction advantage.  $N$ samples of data from a potentially infinite depth QNN made with encoding and function circuits $U_{\text{enc}}$ and $U_{\mathrm{QNN}}$ are provided as input along with quantum and classical methods with associated kernels.  Tests are given as functions of $N$ to emphasize the role of data in the possibility of a prediction advantage.  One can first evaluate a geometric quantity $g_{\mathrm{CQ}}$ that measures the possibility of an advantageous quantum/classical prediction separation without yet considering the actual function to learn.  We show how one can efficiently construct an adversarial function that saturates this limit if the test is passed, otherwise the classical approach is guaranteed to match performance for any function of the data.  To subsequently consider the actual function provided, a label/function specific test may be run using the model complexities $s_C$ and $s_Q$.  If one specifically uses the quantum kernel (QK) method, the red dashed arrows can evaluate if all possible choices of $U_{\mathrm{QNN}}$ lead to an easy classical function for the chosen encoding of the data. 
    \label{fig:FlowChart}
    }
\end{figure*}

At least two paths towards quantum enhancement of machine learning have been considered.  First, motivated by quantum applications in optimization~\cite{Grover:1996,Durr1996Quantum,farhi2001quantum}, the power of quantum computing could, in principle, be used to help improve the training process of existing classical models~\cite{neven2009training,Rebentrost2014}, or enhance inference in graphical models \cite{leifer2008quantum}.  This could include finding better optima in a training landscape or finding optima with fewer queries.  However, without more structure known in the problem, the advantage along these lines may be limited to quadratic or small polynomial speedups \cite{aaronson2009need,mcclean2020low}.

The second vein of interest is the possibility of using quantum models to generate correlations between variables that are inefficient to represent through classical computation. The recent success both theoretically and experimentally for demonstrating quantum computations beyond classical tractability can be taken as evidence that quantum computers can sample from probability distributions that are exponentially difficult to sample from classically~\cite{boixo2018characterizing,arute2019quantum}.  If these distributions were to coincide with real-world distributions, this would suggest the potential for significant advantage.  This is typically the type of advantage that has been sought in recent work on both quantum neural networks~\cite{peruzzo2014variational,mcclean2016theory,farhi2018classification}, which seek to parameterize a distribution through some set of adjustable parameters, and quantum kernel methods~\cite{havlivcek2019supervised} that use quantum computers to define a feature map that maps classical data into the quantum Hilbert space.
The justification for the capability of these methods to exceed classical models often follows similar lines as Refs~\cite{boixo2018characterizing,arute2019quantum} or quantum simulation results.  That is, if the model leverages a quantum circuit that is hard to sample results from classically, then there is potential for a quantum advantage.

In this work, we show quantitatively how this picture is incomplete in machine learning (ML) problems where some training data is provided. The provided data can elevate classical models to rival quantum models, even when the quantum circuits generating the data are hard to compute classically.
We begin with a motivating example and complexity-theoretic argument showing how classical algorithms with data can match quantum output.
Following this, we provide rigorous prediction error bounds for training classical and quantum ML methods based on kernel functions \cite{cortes1995support, scholkopf2002learning, mohri2018foundations, jacot2018neural, novak2019neural, arora2019exact, havlivcek2019supervised, blank2020quantum, bartkiewicz2020experimental, liu2020rigorous} to learn quantum mechanical models. We focus on kernel methods, as they not only provide provable guarantees, but are also very flexible in the functions they can learn. For example, recent advancements in theoretical machine learning show that training neural networks with large hidden layers is equivalent to training an ML model with a particular kernel, known as the neural tangent kernel \cite{jacot2018neural, novak2019neural, arora2019exact}. Throughout, when we refer to classical ML models related to our theoretical developments, we will be referring to ML models that can be easily associated with a kernel, either explicitly as in kernel methods, or implicitly as in the neural tangent kernels.  However, in the numerical section, we will also include performance comparisons to methods where direct  association of a kernel is challenging, such as random forest methods. In the quantum case, we will also show how quantum ML based on kernels can be made equivalent to training an infinite depth quantum neural network.

We use our prediction error bounds to devise a flowchart for testing potential quantum prediction advantage, the separation between prediction errors of quantum and classical ML models for a fixed amount of training data.  
The most important test is a geometric difference between kernel functions defined by classical and quantum ML. Formally, the geometric difference is defined by the closest efficient classical ML model. In practice, one should consider the geometric difference with respect to a suite of optimized classical ML models. If the geometric difference is small, then a classical ML method is guaranteed to provide similar or better performance in prediction on the data set, independent of the function values or labels.  Hence this represents a powerful, function independent pre-screening that allows one to evaluate if there is any possibility of better performance.
On the other hand, if the geometry differs greatly, we show both the existence of a data set that exhibits large prediction advantage using the quantum ML model and how one can construct it efficiently.  While the tools we develop could be used to compare and construct hard classical models like hash functions, we enforce restrictions that allow us to say something about a quantum separation.  In particular, the feature map will be white box, in that a quantum circuit specification is available for the ideal feature map, and that feature map can be made computationally hard to evaluate classically.  A constructive example of this is a discrete log feature map, where a provable separation for our kernel is given in Appendix~\ref{app:qadv}.  Additionally, the minimum over classical models means that classical hash functions are reproduced formally by definition.

Moreover, application of these tools to existing models in the literature rules many of them out immediately, providing a powerful sieve for focusing development of new data encodings.
Following these constructions, in numerical experiments, we find that a variety of common quantum models in the literature perform similarly or worse than classical ML on both classical and quantum data sets due to a small geometric difference.
The small geometric difference is a consequence of the exponentially large Hilbert space employed by existing quantum models, where all inputs are too far apart.
To circumvent the setback, we propose an improvement, which enlarges the geometric difference by projecting quantum states embedded from classical data back to approximate classical representation \cite{Huang2020, cotler2020quantum, paini2019approximate}.
With the large geometric difference endowed by the projected quantum model, we are able to construct engineered data sets to demonstrate large prediction advantage over \emph{common} classical ML models in numerical experiments up to $30$ qubits. Despite our constructions being based on methods with associated kernels, we find empirically that the prediction advantage remains robust across tested classical methods, including those without an easily determined kernel. This opens the possibility to use a small quantum computer to generate efficiently verifiable machine learning problems that could be challenging for classical ML models.

\vspace{-1em}
\section*{Results}
\subsection{Setup and motivating example}
We begin by setting up the problems and methods of interest for classical and quantum models, and then provide a simple motivating example for studying how data can increase the power of classical models on quantum data.  The focus will be a supervised learning task with a collection of $N$ training examples $\{(x_i, y_i)\}$, where $x_i$ is the input data and $y_i$ is an associated label or value. We assume that $x_i$ are sampled independently from a data distribution $\mathcal{D}$.

In our theoretical analysis, we will consider $y_i \in \mathbb{R}$ to be generated by some quantum model.
In particular, we consider a continuous encoding unitary that maps classical vector $x_i$ into quantum state $\ket{x_i} = U_{\text{enc}}(x_i)\ket{0}^{\otimes n}$ and refer to the corresponding density matrix as $\rho(x_i)$.  The expressive power of these embeddings have been investigated from a functional analysis point of view~\cite{lloyd2020quantum,schuld2020effect}, however the setting where data is provided requires special attention. 
The encoding unitary is followed by a unitary $U_{\text{QNN}}(\theta)$.
We then measure an observable $O$ after the quantum neural network.
This produces the label/value for input $x_i$ given as $y_i = f(x_i) = \bra{x_i} U_{\text{QNN}}^\dagger O U_{\text{QNN}} \ket{x_i}$.
The quantum model considered here is also referred to as a quantum neural network (QNN) in the literature \cite{farhi2018classification, mcclean2018barren}.
The goal is to understand when it is easy to predict the function $f(x)$ by training classical/quantum machine learning models.

With notation in place, we turn to a simple motivating example to understand how the availability of data in machine learning tasks can change computational hardness.
Consider data points $\{x_i\}_{i=1}^N$ that are $p$-dimensional classical vectors with $\norm{x_i}_2 = 1$, and use amplitude encoding \cite{grant2019initialization, schuld2020circuit, larose2020robust} to encode the data into an $n$-qubit state $\ket{x_i} = \sum_{k=1}^p x_{i}^k \ket{k}$.
If $U_{\text{QNN}}$ is a time-evolution under a many-body Hamiltonian, then the function $f(x) = \bra{x} U_{\text{QNN}}^\dagger O U_{\text{QNN}} \ket{x}$ is in general hard to compute classically~\cite{harrow2017quantum}
, even for a single input state.  In particular, we have the following proposition showing that if a classical algorithm can compute $f(x)$ efficiently, then quantum computers will be no more powerful than classical computers; see Appendix~\ref{app:motivating-example} for a proof.
\begin{proposition}\label{prop:BQP=BPP}
If a classical algorithm without training data can compute $f(x)$ efficiently for any $U_{\text{QNN}}$ and $O$, then BPP=BQP.
\end{proposition}
Nevertheless, it is incorrect to conclude that training a classical model from data to learn this evolution is hard.
To see this, we write out the expectation value as
\begin{align}
    f(x_i) &= \left( \sum_{k=1}^p x_{i}^{k*} \bra{k} \right) U_{\text{QNN}}^\dagger O U_{\text{QNN}} \left( \sum_{l=1}^p x_{i}^l \ket{l} \right) \notag \\
    &=  \sum_{k=1}^p \sum_{l=1}^p B_{kl} x_{i}^{k*} x_{i}^l,
\end{align}
which is a quadratic function with $p^2$ coefficients $B_{kl}=\bra{k}U_{\text{QNN}}^\dagger O U_{\text{QNN}}\ket{l}$.
Using the theory developed later in this work, we can show that, for any $U_{\text{QNN}}$ and $O$, training a specific classical ML model on a collection of $N$ training examples $\{(x_i, y_i = f(x_i))\}$ would give rise to a prediction model $h(x)$ with
\begin{equation}
    \E_{x \sim \mathcal{D}} |h(x) - f(x)| \leq c \sqrt{\frac{p^2}{N}},
\end{equation}
for a constant $c > 0$. We refer to Appendix~\ref{app:motivating-example} for the proof of this result.
Hence, with $N \propto p^2 / \epsilon^2$ training data, one can train a classical ML model to predict the function $f(x)$ up to an additive prediction error $\epsilon$. This elevation of classical models through some training samples is illustrative of the power of data.
In Appendix~\ref{sec:powerdatarigor}, we give a rigorous complexity-theoretic argument on the computational power provided by data. A cartoon depiction of the complexity separation induced by data is provided in Fig.~\ref{fig:FlowChart}(a).

While this simple example makes the basic point that sufficient data can change complexity considerations, it perhaps opens more questions than it answers.  For example, it uses a rather weak encoding into amplitudes and assumes one has access to an amount of data that is on par with the dimension of the model.  The more interesting cases occur if we strengthen the data encoding, include modern classical ML models, and consider number of data $N$ much less than the dimension of the model.  These more interesting cases are the ones we quantitatively answer.

Our primary interest will be ML algorithms that are much stronger than fitting a quadratic function and the input data is provided in more interesting ways than an amplitude encoding.
In this work, we focus on both classical and quantum ML models based on kernel functions $k(x_i, x_j)$.
At a high level, a kernel function can be seen as a measure of similarity, if $k(x_i, x_j)$ is large when $x_i$ and $x_j$ are close.
When considered for finite input data, a kernel function may be represented as a matrix $K_{ij} = k(x_i, x_j)$ and the conditions required for kernel methods are satisfied when the matrix representation is Hermitian and positive semi-definite. 
\begin{figure}[t]
\centering
\includegraphics[width=1.0\columnwidth]{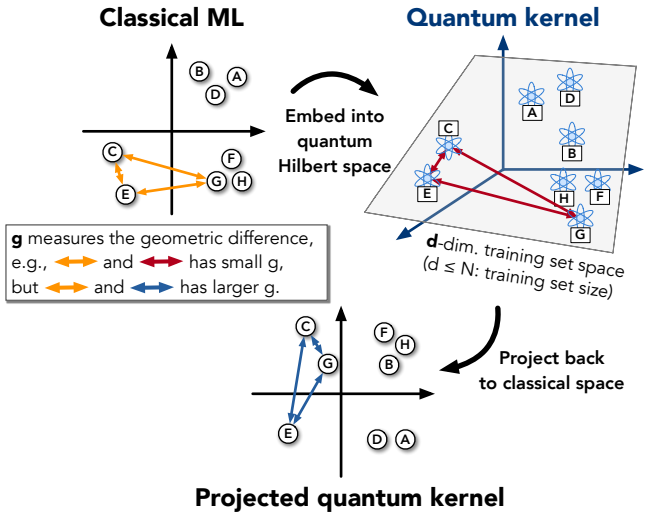}
    \caption{ Cartoon of the geometry (kernel function) defined by classical and quantum ML models. The letters A, B, ... represent data points $\{x_i\}$ in different spaces with arrows representing the similarity measure (kernel function) between data.  The geometric difference $g$ is a difference between similarity measures (arrows) in different ML models and $d$ is an effective dimension of the data set in the quantum Hilbert space.
    \label{fig:geometry}}
\end{figure}

A given kernel function corresponds to a nonlinear feature mapping $\phi(x)$ that maps $x$ to a possibly infinite-dimensional feature space, such that $k(x_i, x_j) = \phi(x_i)^\dagger \phi(x_j)$. This is the basis of the so-called ``kernel trick'' where intricate and powerful maps $\phi(x_i)$ can be implemented through the evaluation of relatively simple kernel functions $k$.  As a simple case, in the example above, using a kernel of $k(x_i, x_j)=|\braket{x_i}{x_j}|^2$ corresponds to a feature map $\phi(x_i) = \sum_{kl} x_i^{k*} x_i^l \ket{k} \otimes \ket{l}$ which is capable of learning quadratic functions in the amplitudes.  In kernel based ML algorithms, the trained model can always be written as $h(x) = w^\dagger \phi(x)$ where $w$ is a vector in the feature space defined by the kernel.
For example, training a convolutional neural network with large hidden layers \cite{jacot2018neural,li2019enhanced} is equivalent to using a corresponding neural tangent kernel $k^{\text{CNN}}$. The feature map $\phi^{\text{CNN}}$ for the kernel $k^{\text{CNN}}$ is a nonlinear mapping that extracts all local properties of $x$ \cite{li2019enhanced}.
In quantum mechanics, similarly a kernel function can be defined using the native geometry of the quantum state space $\ket{x}$. For example, we can define the kernel function as $\braket{x_i}{x_j}$ or $|\braket{x_i}{x_j}|^2$.
Using the output from this kernel in a method like a classical support vector machine \cite{cortes1995support} defines the quantum kernel method. 

A wide class of functions can be learned with a sufficiently large amount of data by using the right kernel function $k$. 
For example, in contrast to the perhaps more natural kernel, $\braket{x_i}{x_j}$, the quantum kernel $k^{\mathrm{Q}}(x_i, x_j) = |\braket{x_i}{x_j}|^2 = \tr(\rho(x_i) \rho(x_j))$ can learn arbitrarily deep quantum neural network $U_{\text{QNN}}$ that measures any observable $O$ (shown in Appendix~\ref{sec:qk_qnn}), and the Gaussian kernel, $k^\gamma(x_i, x_j) = \exp(-\gamma ||x_i - x_j||^2)$ with hyper-parameter $\gamma$, can learn any continuous function in a compact space \cite{micchelli2006universal}, which includes learning any QNN.
Nevertheless, the required amount of data $N$ to achieve a small prediction error could be very large in the worst case.
Although we will work with other kernels defined through a quantum space, due both to this expressive property and terminology of past work, we will refer to $k^{\mathrm{Q}}(x_i, x_j) = \tr\left[ \rho(x_i) \rho(x_j) \right]$ as the quantum kernel method throughout this work, which is also the definition given in \cite{havlivcek2019supervised}.  

\subsection{Testing quantum advantage} \label{sec:tests} 
We now construct our more general framework for assessing the potential for quantum prediction advantage in a machine learning task.  Beginning from a general result, we build both intuition and practical tests based on the geometry of the learning spaces.  This framework is summarized in Fig.~\ref{fig:FlowChart}.

Our foundation is a general prediction error bound for training classical/quantum ML models to predict some quantum model defined by $f(x) = \Tr( O^U \rho(x))$ derived from concentration inequalities, where $O^U = U_{\text{QNN}}^\dagger O U_{\text{QNN}}$.
Suppose we have obtained $N$ training examples $\{(x_i, y_i = f(x_i))\}$.
After training on this data, there exists an ML algorithm that outputs $h(x) = w^\dagger \phi(x)$ using kernel $k(x_i, x_j) = K_{ij} = \phi(x_i)^\dagger \phi(x_j)$ which has a simplified prediction error bounded by
\begin{align}
    \mathbb{E}_{x \sim \mathcal{D}}|h(x) - f(x)| & \leq c \sqrt{\frac{s_K(N)}{N}} \label{eq:main_pred}
\end{align}
for a constant $c > 0$ and $N$ independent samples from the data distribution $\mathcal{D}$. We note here that this and all subsequent bounds have a key dependence on the quantity of data $N$, reflecting the role of data to improve prediction performance.  Due to a scaling freedom between $\alpha \phi(x)$ and $w / \alpha$, we have assumed $\sum_{i=1}^N \phi(x_i)^\dagger\phi(x_i) = \Tr(K) = N$. A derivation of this result is given in Appendix~\ref{sec:proofmain}.

Given this core prediction error bound, we now seek to understand its implications.
The main quantity that determines the prediction error is
\begin{equation}
    s_K(N) = \sum_{i=1}^N \sum_{j=1}^N (K^{-1})_{ij} \Tr( O^U \rho(x_i)) \Tr( O^U \rho(x_j)).
\end{equation}
The quantity $s_K(N)$ is equal to the model complexity of the trained function $h(x) = w^\dagger \phi(x)$, where $s_K(N) = \norm{w}^2 = w^\dagger w$ after training. A smaller value of $s_K(N)$ implies better generalization to new data $x$ sampled from the distribution $\mathcal{D}$.
Intuitively, $s_K(N)$ measures whether the closeness between $x_i, x_j$ defined by the kernel function $k(x_i, x_j)$ matches well with the closeness of the observable expectation for the quantum states $\rho(x_i), \rho(x_j)$, recalling that a larger kernel value indicates two points are closer.
The computation of $s_K(N)$ can be performed efficiently on a classical computer by inverting an $N\times N$ matrix $K$ after obtaining the $N$ values $\Tr(O^U \rho(x_i))$ by performing order $N$ experiments on a physical quantum device. The time complexity scales at most as order $N^3$.
Due to the connection between $w^\dagger w$ and the model complexity, a regularization term $w^\dagger w$ is often added to the optimization problem during the training of $h(x) = w^\dagger \phi(x)$, see e.g., \cite{krogh1992simple, cortes1995support, suykens1999least}. Regularization prevents $s_K(N)$ from becoming too large at the expense of not completely fitting the training data. A detailed discussion and proof under regularization is given in Appendix~\ref{sec:proofmain}~and~\ref{app:sdgdetail}.

The prediction error upper bound can often be shown to be asymptotically tight by proving a matching lower bound. As an example, when $k(x_i, x_j)$ is the quantum kernel $\Tr(\rho(x_i) \rho(x_j))$, we can deduce that $s_K(N) \leq \Tr(O^2)$ hence one would need a number of data $N$ scaling as $\Tr(O^2)$. In Appendix~\ref{app:lowerbound}, we give a matching lower bound showing that a scaling of $\Tr(O^2)$ is unavoidable if we assume a large Hilbert space dimension. This lower bound holds for any learning algorithm and not only for quantum kernel methods. The lower bound proof uses mutual information analysis and could easily extend to other kernels. This proof strategy is also employed extensively in a follow-up work \cite{huang2021informationtheoretic} to devise upper and lower bounds for classical and quantum ML in learning quantum models. Furthermore, not only are the bounds asymptotically tight, in numerical experiments given in Appendix~\ref{app:additional} we find that the prediction error bound also captures the performance of other classical ML models not based on kernels where the constant factors are observed to be quite modest.

Given some set of data, if $s_K(N)$ is found to be small relative to ${N}$ after training for a classical ML model, this quantum model $f(x)$ can be predicted accurately even if $f(x)$ is hard to compute classically for any given $x$. In order to formally evaluate the potential for quantum prediction advantage generally, one must take $s_K(N)$ to be the minimal over efficient classical models. However, we will be more focused on minimally attainable values over a reasonable set of classical methods with tuned hyperparameters.  This prescribes an effective method for evaluating potential quantum advantage in practice, and already rules out a considerable number of examples from the literature.

From the bound, we can see that the potential advantage for one ML algorithm defined by $K^1$ to predict better than another ML algorithm defined by $K^2$ depends on the largest possible separation between $s_{K^1}$ and $s_{K^2}$ for a data set.
The separation can be characterized by defining an asymmetric geometric difference that depends on the dataset, but is independent of the function values or labels.  Hence evaluating this quantity is a good first step in understanding if there is a potential for quantum advantage, as shown in Fig.~\ref{fig:FlowChart}.  This quantity is defined by
\begin{align}
   g_{12} =  g(K^1 || K^2) = \sqrt{\norm{\sqrt{K^2} (K^1)^{-1} \sqrt{K^2}}_\infty},
\end{align}
where $\norm{.}_\infty$ is the spectral norm of the resulting matrix and we assume $\Tr(K^1) = \Tr(K^2) = N$. One can show that
$s_{K^1} \leq g_{12}^2 s_{K^2}$, which implies the prediction error bound $c \sqrt{s_{K^1} / N} \leq c g_{12} \sqrt{s_{K^2} / N}$.
A detailed derivation is given in Appendix~\ref{app:geo-detail} and an illustration of $g_{12}$ can be found in Fig.~\ref{fig:geometry}.
The geometric difference $g(K^1 || K^2)$ can be computed on a classical computer by performing a singular value decomposition of the $N\times N$ matrices $K^1$ and $K^2$. Standard numerical analysis packages \cite{LAPACK} provide highly efficient computation of a singular value decomposition in time at most order $N^3$.
Intuitively, if $K^1(x_i, x_j)$ is small/large when $K^2(x_i, x_j)$ is small/large, then the geometric difference $g_{12}$ is a small value $\sim 1$, where $g_{12}$ grows as the kernels deviate.

To see more explicitly how the geometric difference allows one to make statements about the possibility for one ML model to make different predictions from another,  consider the geometric difference $g_{\mathrm{CQ}} = g(K^{\mathrm{C}} || K^{\mathrm{Q}})$ between a classical ML model with kernel $k^{\mathrm{C}}(x_i, x_j)$ and a quantum ML model, e.g., with $k^{\mathrm{Q}}(x_i, x_j) = \Tr(\rho(x_i) \rho(x_j))$.  If $g_{\mathrm{CQ}}$ is small, because
\begin{equation}
s_{\mathrm{C}} \leq g_{\mathrm{CQ}}^2 s_{\mathrm{Q}},    
\end{equation}
the classical ML model will always have a similar or better model complexity $s_K(N)$ compared to the quantum ML model.
This implies that the prediction performance for the classical ML will likely be competitive or better than the quantum ML model, and one is likely to prefer using the classical model. This is captured in the first step of our flowchart in Fig.~\ref{fig:FlowChart}.

In contrast, if $g_{\mathrm{CQ}}$ is large we show that there exists a data set with $s_{\mathrm{C}} = g_{\mathrm{CQ}}^2 s_{\mathrm{Q}}$ with the quantum model exhibiting superior prediction performance. An efficient method to explicitly construct such a maximally divergent data set is given in Appendix~\ref{app:engdata} and a numerical demonstration of the stability of this separation is provided in the next section.
While a formal statement about classical methods generally requires defining it over all efficient classical methods, in practice, we consider $g_{\mathrm{CQ}}$ to be the minimum geometric difference among a suite of optimized classical ML models.  Our engineered approach minimizes this value as a hyperparameter search to find the best classical adversary, and shows remarkable robustness across classical methods including those without an associated kernel, such as random forests \cite{breiman2001random}.

\begin{figure*}[t]
\centering
\includegraphics[width=1.0\textwidth]{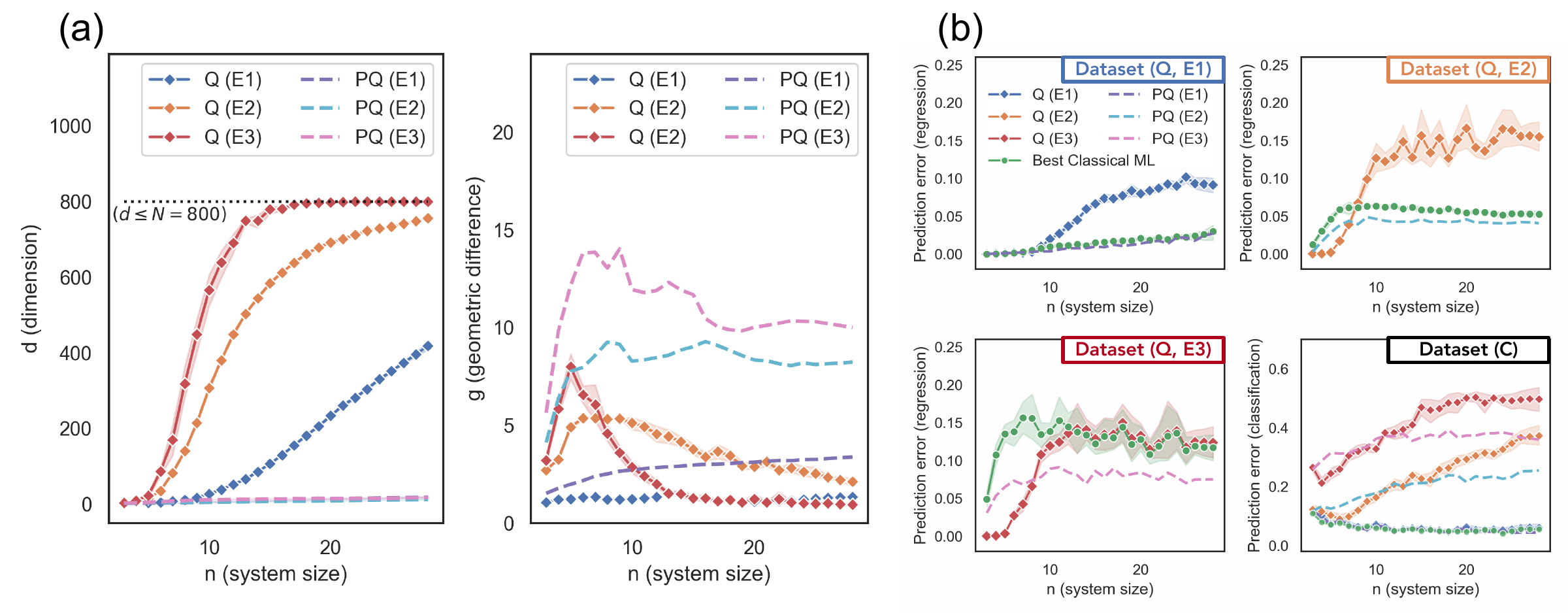}
    \caption{Relation between dimension $d$, geometric difference $g$, and prediction performance. The shaded regions are the standard deviation over $10$ independent runs and $n$ is the number of qubits in the quantum encoding and dimension of the input for the classical encoding.
    (a)~The approximate dimension $d$ and the geometric difference $g$ with classical ML models for quantum kernel (Q) and projected quantum kernel (PQ) under different embeddings and system sizes $n$.
    (b)~Prediction error (lower is better) of the quantum kernel method (Q), projected quantum kernel method (PQ), and classical ML models on classical~(C) and quantum~(Q) data sets with number of data $N = 600$. As $d$ grows too large, the geometric difference $g$ for quantum kernel becomes small. We see that small geometric difference $g$ always results in classical ML being competitive or outperforming the quantum ML model. When $g$ is large, there is a potential for improvement over classical ML. For example, projected quantum kernel improves upon the best classical ML in Dataset (Q, E3). \label{fig:DimensionAccuracy}}
\end{figure*}

In the specific case of the quantum kernel method with $K^Q_{ij} = k^{\mathrm{Q}}(x_i, x_j) = \Tr(\rho(x_i) \rho(x_j))$, we can gain additional insights into the model complexity $s_K$, and sometimes make conclusions about classically learnability for all possible $U_{\mathrm{QNN}}$ for the given encoding of the data. Let us define $\mathrm{vec}(X)$ for a Hermitian matrix $X$ to be a vector containing the real and imaginary part of each entry in $X$.
In this case, we find $s_{Q} = \mathrm{vec}(O^U)^T P_Q \mathrm{vec}(O^U)$, where $P_Q$ is the projector onto the subspace formed by $\{\mathrm{vec}(\rho(x_1)), \ldots, \mathrm{vec}(\rho(x_N))\}$. 
We highlight
\begin{align}
    d = \text{dim}(P_Q) = \text{rank}(K^{\mathrm{Q}}) \leq N,
\end{align}
which defines the effective dimension of the quantum state space spanned by the training data.
An illustration of the dimension $d$ can be found in Fig.~\ref{fig:FlowChart}.
Because $P_Q$ is a projector and has eigenvalues $0$~or~$1$,
$s_{\mathrm{Q}} \leq \min(d, \mathrm{vec}(O^U)^T \mathrm{vec}(O^U)) = \min(d, \Tr(O^2))$ assuming $\norm{O}_\infty \leq 1$.
Hence in the case of the quantum kernel method, the prediction error bound may be written as
\begin{equation} \label{eq:QKprederr}
    \mathbb{E}_{x \in \mathcal{D}} |h(x) - f(x)| \leq c \sqrt{\frac{\min(d, \Tr(O^2))}{N}}.
\end{equation}
A detailed derivation is given in Appendix~\ref{sec:QKmethodanalysis}.
We can also consider the approximate dimension $d$, where small eigenvalues in $K^{\mathrm{Q}}$ are truncated, by incurring a small training error.
After obtaining $K^{\mathrm{Q}}$ from a quantum device, the dimension $d$ can be computed efficiently on a classical machine by performing a singular value decomposition on the $N \times N$ matrix $K^{\mathrm{Q}}$. Estimation of $\Tr(O^2)$ can be performed by sampling random states $\ket{\psi}$ from a quantum $2$-design, measuring $O$ on $\ket{\psi}$, and performing statistical analysis on the measurement data \cite{Huang2020}.
This prediction error bound shows that a quantum kernel method can learn any $U_{\mathrm{QNN}}$ when the dimension of the training set space $d$ or the squared Frobenius norm of observable $\Tr(O^2)$ is much smaller than the amount of data $N$.
In Appendix~\ref{app:lowerbound}, we show that quantum kernel methods are optimal for learning quantum models with bounded $\Tr(O^2)$ as they saturate the fundamental lower bound.
However, in practice, most observables, such as Pauli operators, will have exponentially large $\Tr(O^2)$, so the central quantity is the dimension $d$. Using the prediction error bound for the quantum kernel method, if both $g_{\mathrm{CQ}}$ and $\min(d, \Tr(O^2))$ are small, then a classical ML would also be able to learn any $U_{\mathrm{QNN}}$.  In such a case, one must conclude that the given encoding of the data is classically easy, and this cannot be affected by an arbitrarily deep $U_{\mathrm{QNN}}$.
This constitutes the bottom left part of our flowchart in Fig.~\ref{fig:FlowChart}.

Ultimately, to see a prediction advantage in a particular data set with specific function values/labels, we need a large separation between $s_{\mathrm{C}}$ and $s_{\mathrm{Q}}$. This happens when the inputs $x_i, x_j$ considered close in a quantum ML model are actually close in the target function $f(x)$, but are far in classical ML.  This is represented as the final test in Fig.~\ref{fig:FlowChart} and the methodology here outlines how this result can be achieved in terms of its more essential components.

\subsection{Projected quantum kernels}

In addition to analyzing existing quantum models, the analysis approach introduced also provides suggestions for new quantum models with improved properties, which we now address here.  For example, if we start with the original quantum kernel, when the effective dimension $d$ is large, kernel $\Tr(\rho(x_i) \rho(x_j))$, which is based on a fidelity-type metric, will regard all data to be far from each other and the kernel matrix $K^{\mathrm{Q}}$ will be close to identity. This results in a small geometric difference $g_{\mathrm{CQ}}$ leading to classical ML models being competitive or outperforming the quantum kernel method. In Appendix~\ref{app:limitqk}, we present a simple quantum model that requires an exponential amount of samples to learn using the quantum kernel $\Tr(\rho(x_i) \rho(x_j))$, but only needs a linear number of samples to learn using a classical ML model.

To circumvent this setback, we propose a family of projected quantum kernels as a solution.
These kernels work by projecting the quantum states to an approximate classical representation, e.g., using reduced physical observables or classical shadows \cite{gosset2018compressed, aaronson2020shadow, aaronson2019gentle, paini2019approximate, Huang2020}.
Even if the training set space has a large dimension $d \sim N$, the projection allows us to reduce to a low-dimensional classical space that can generalize better. Furthermore, by going through the exponentially large quantum Hilbert space, the projected quantum kernel can be challenging to evaluate without a quantum computer. In numerical experiments, we find that the classical projection increases rather than decreases the geometric difference with classical ML models. These constructions will be the foundation of our best performing quantum method later.

One of the simplest forms of projected quantum kernel is to measure the one-particle reduced density matrix (1-RDM) on all qubits for the encoded state, $\rho_k(x_i) = \tr_{j \neq k}[\rho(x_i)]$, then define the kernel as
\begin{align}
    k^{\text{PQ}}(x_i, x_j) = \exp \left(-\gamma  \sum_k \norm{\rho_k(x_i) - \rho_k(x_j)}_F^2 \right).
\end{align}
This kernel defines a feature map function in the 1-RDM space that is capable of expressing arbitrary functions of powers of the 1-RDMs of the quantum state.  From non-intuitive results in density functional theory, we know even one body densities can be sufficient for determining exact ground state~\cite{hohenberg1964inhomogeneous} and time-dependent~\cite{runge1984density} properties of many-body systems under modest assumptions.
In Appendix~\ref{app:pqk}, we provide examples of other projected quantum kernels. This includes an efficient method for computing a kernel function that contains all orders of RDMs using local randomized measurements and the formalism of classical shadows \cite{Huang2020}. The classical shadow formalism allows efficient construction of RDMs from very few measurements. In Appendix~\ref{app:qadv}, we show that projected versions of quantum kernels lead to a simple and rigorous quantum speed-up in a recently proposed learning problem based on discrete logarithms \cite{liu2020rigorous}.

\subsection{Numerical studies}

\begin{figure*}[t]
\centering
\includegraphics[width=0.86\textwidth]{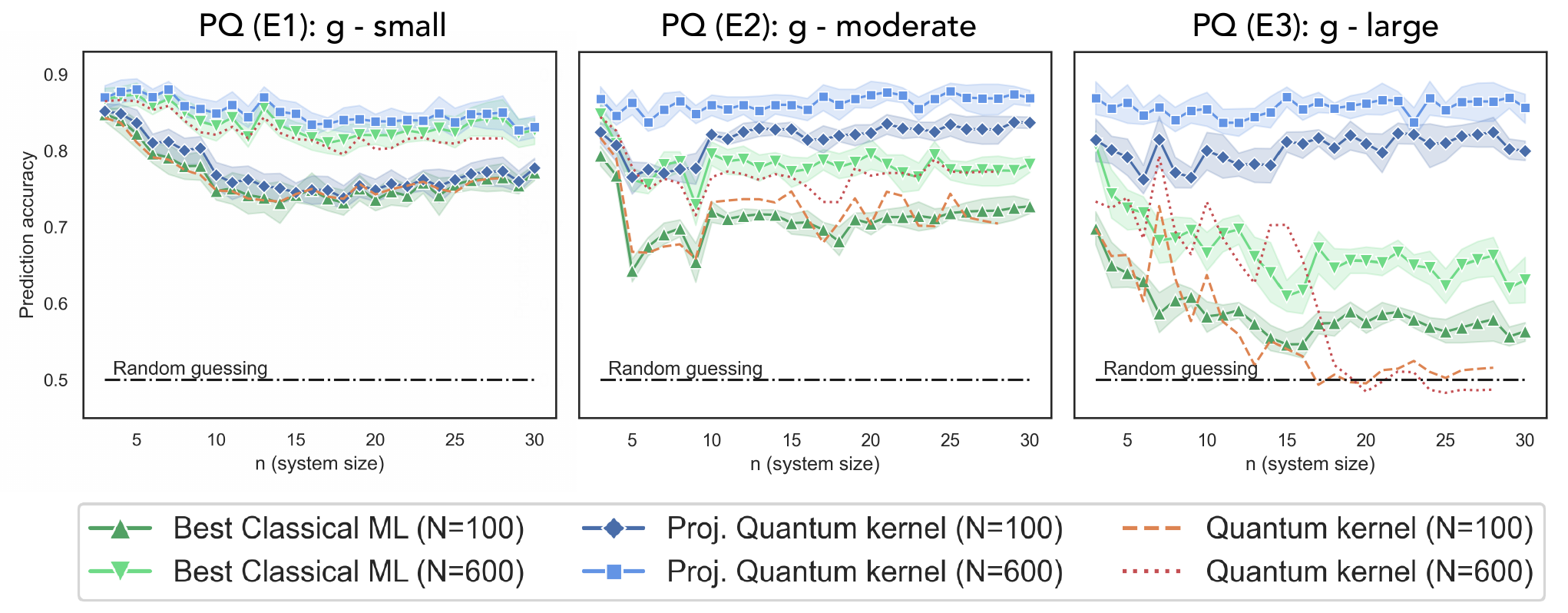}
    \caption{Prediction accuracy (higher the better) on engineered data sets. A label function is engineered to match the geometric difference $g(\mathrm{C} || \mathrm{PQ})$ between projected quantum kernel and classical approaches, demonstrating a significant gap between quantum and the best classical models up to 30 qubits when $g$ is large. We consider the best performing classical ML models among Gaussian SVM, linear SVM, Adaboost, random forest, neural networks, and gradient boosting. We only report the accuracy of the quantum kernel method up to system size $n = 28$ due to the high simulation cost and the inferior performance. \label{fig:QuantumError}}
\end{figure*}

We now provide numerical evidence up to 30 qubits that supports our theory on the relation between the dimension $d$, the geometric difference $g$, and the prediction performance. Using the projected quantum kernel, the geometric difference $g$ is much larger and we see the strongest empirical advantage of a scalable quantum model on quantum data sets to date.  These are the largest combined simulation and analysis in digital quantum machine learning that we are aware of, and make use of the TensorFlow and TensorFlow-Quantum package~\cite{broughton2020tensorflow}, reaching a peak throughput of up to 1.1 quadrillion floating point operations per second (petaflop/s). Trends of  approximately 300 teraflop/s for quantum simulation and 800 teraflop/s for classical analysis were observed up to the maximum experiment size with the overall floating point operations across all experiments totalling approximately 2 quintillion (exaflop).

In order to mimic a data distribution that pertains to real-world data, we conduct our experiments around the fashion-MNIST data set~\cite{xiao2017fashion}, which is an image classification for distinguishing clothing items, and is more challenging than the original digit-based MNIST source \cite{lecun2010mnist}. We pre-process the data using principal component analysis \cite{jolliffe1986principal} to transform each image into an $n$-dimensional vector. The same data is provided to the quantum and classical models, where in the classical case the data is the $n$-dimensional input vector, and the quantum case uses a given circuit to embed the $n$-dimensional vector into the space of $n$ qubits.  For quantum embeddings, we explore three options, E1 is a separable rotation circuit~\cite{schuld2019quantum, schuld2020circuit, skolik2020layerwise}, E2 is an IQP-type embedding circuit~\cite{havlivcek2019supervised}, and E3 is a Hamiltonian evolution circuit, with explicit constructions in Appendix~\ref{app:detailnum}.

For the classical ML task (C), the goal is to correctly identify the images as shirts or dresses from the original data set.  For the quantum ML tasks, we use the same fashion-MINST source data and embeddings as above, but take as function values the expectation value of a local observable that has been evolved under a quantum neural network resembling the Trotter evolution of 1D-Heisenberg model with random couplings.
In these cases, the embedding is taken as part of the ground truth, so the resulting function will be different depending on the quantum embedding.
For these ML tasks, we compare against the best performing model from a list of standard classical ML algorithms with properly tuned hyper-parameters (see Appendix~\ref{app:detailnum} for details).

In Fig.~\ref{fig:DimensionAccuracy}, we give a comparison between the prediction performance of classical and quantum ML models. One can see that not only do classical ML models perform best on the original classical dataset, the prediction performance for the classical methods on the quantum datasets is also very competitive and can even outperform existing quantum ML models despite the quantum ML models having access to the training embedding while the classical methods do not. The performance of the classical ML model is especially strong on Dataset (Q, E1) and Dataset (Q, E2).  This elevation of the classical performance is evidence of the power of data.
Moreover, this intriguing behavior and the lack of quantum advantage may be explained by considering the effective dimension $d$ and the geometric difference $g$ following our theoretical constructions.
From Fig.~\ref{fig:DimensionAccuracy}a, we can see that the dimension $d$ of the original quantum state space grows rather quickly, and the geometric difference $g$ becomes small as the dimension becomes too large ($d \propto N$) for the standard quantum kernel.  The saturation of the dimension coincides with the decreasing and statistical fluctuations in performance seen in Fig.~\ref{fig:QuantumError}.
Moreover, given poor ML performance a natural instinct is to throw more resources at the problem, e.g. more qubits, but as demonstrated here, doing this for naïve quantum kernel methods is likely to lead to tiny inner products and even worse performance.
In contrast, the projected quantum space has a low dimension even when $d$ grows, and yields a higher geometric difference $g$ for all embeddings and system sizes.
Our methodology predicts that, when $g$ is small, classical ML model will be competitive or outperform the quantum ML model. This is verified in Fig.~\ref{fig:DimensionAccuracy}b for both the original and projected quantum kernel, where a small geometric difference $g$ leads to a very good performance of classical ML models and no large quantum advantage can be seen.
Only when the geometric difference $g$ is large (projected kernel method with embedding E3) can we see some mild advantage over the best classical method.
This result holds disregarding any detail of the quantum evolution we are trying to learn, even for ones that are hard to simulate classically.


In order to push the limits of separation between quantum and classical approaches in a learning setting, we now consider a set of engineered data sets with function values designed to saturate the geometric inequality $s_{\mathrm{C}} \leq g(K^\mathrm{C} || K^\mathrm{PQ})^2 s_{\mathrm{PQ}}$ between classical ML models with associated kernels and the projected quantum kernel method. In particular, we design the data set such that $s_{\mathrm{PQ}}=1$ and $s_{\mathrm{C}} = g(K^\mathrm{C} || K^\mathrm{PQ})^2$. Recall from Eq.~\eqref{eq:main_pred}, this data set will hence show the largest separation in the prediction error bound $\sqrt{s(N) / N}$. The engineered data set is constructed via a simple eigenvalue problem with the exact procedure described in Appendix ~\ref{app:engdata} and the results are shown in Fig.\ref{fig:QuantumError}. As the quantum nature of the encoding increases from E1 to E3, corresponding to increasing $g$, the performance of both the best classical methods and the original quantum kernel decline precipitously.  The advantage of projected quantum kernel closely follows the geometric difference $g$ and reaches more than $20\%$ for large sizes. Despite the optimization of $g$ only being possible for classical methods with an associated kernel, the performance advantage remains stable across other common classical methods.
Note that we also constructed engineered data sets saturating the geometric inequality between classical ML and the original quantum kernel, but the small geometric difference $g$ presented no empirical advantage at large system size (see Appendix~\ref{app:additional}).

In keeping with our arguments about the role of data, when we increase the number of training data $N$, all methods improve, and the advantage will gradually diminish.
While this data set is engineered, it shows the strongest empirical separation on the largest system size to date.  We conjecture that this procedure could be used with a quantum computer to create challenging data sets that are easy to learn with a quantum device, hard to learn classically, while still being easy to verify classically given the correct labels.  Moreover, the size of the margin implies that this separation may even persist under moderate amounts of noise in a quantum device.

\section*{Discussion}
The use of quantum computing in machine learning remains an exciting prospect, but quantifying quantum advantage for such applications has some subtle issues that one must approach carefully.  Here, we constructed a foundation for understanding opportunities for quantum advantage in a learning setting. We showed quantitatively how classical ML algorithms with data can become computationally more powerful, and a prediction advantage for quantum models is not guaranteed even if the data comes from a quantum process that is challenging to independently simulate. Motivated by these tests, we introduced projected quantum kernels. On engineered data sets, projected quantum kernels outperform all tested classical models in prediction error. To the authors' knowledge, this is the first empirical demonstration of such a large separation between quantum and classical ML models.
 
This work suggests a simple guidebook for generating ML problems which give a large separation between quantum and classical models, even at a modest number of qubits. The size of this separation and trend up to 30 qubits suggests the existence of learning tasks that may be easy to verify, but hard to model classically, requiring just a modest number of qubits and allowing for device noise.  Claims of true advantage in a quantum machine learning setting require not only benchmarking classical machine learning models, but also classical approximations of quantum models.  Additional work will be needed to identify embeddings that satisfy the sometimes conflicting requirements of being hard to approximate classically and exhibiting meaningful signal on local observables for very large numbers of qubits.  Further research will be required to find use cases on data sets closer to practical interest and evaluate potential claims of advantage, but we believe the tools developed in this work will help to pave the way for this exciting frontier. 

\section*{Acknowledgements}
The authors want to thank Richard Kueng, John Platt, John Preskill, Thomas Vidick, Nathan Wiebe, and Chun-Ju Wu for valuable inputs and inspiring discussions.  We thank B\'alint Pat\'o for crucial contributions in setting up simulations.

\onecolumngrid

\bibliographystyle{apsrev4-1_with_title}
\bibliography{references}

\newpage
\appendix

\section{Rigorous proofs for statements regarding the motivating example}
\label{app:motivating-example}

We first give a simple proof that the motivating example $f(x)$ considered in the main text is in general hard to compute classically. Then, we show that training a classical ML model to predict the function $f(x)$ is easy on a classical computer.

\begin{proposition}[Restatement of Proposition~\ref{prop:BQP=BPP}]
Consider input vector $x \in \mathbb{R}^p$ encoded into an $n$-qubit state $\ket{x} = \sum_{k=1}^p x_k \ket{k}$.
If a randomized classical algorithm can compute
\begin{equation}
f(x) = \bra{x} U_{\text{QNN}}^\dagger O U_{\text{QNN}} \ket{x}    
\end{equation}
up to $0.15$-error with high probability over the randomness in the classical algorithm for any $n$, $U_{\text{QNN}}$ and $O$ in a time polynomial to the description length of $U_{\text{QNN}}$ and $O$, the input vector size $p$, and the qubit system size $n$, then
\begin{equation}
    \text{BPP} = \text{BQP}.
\end{equation}
\end{proposition}
\begin{proof}
We consider $p = 1$ and $\ket{x} = \ket{0^n}$ the all zero computational basis state.
A language $L$ is in BQP if and only if there exists a polynomial-time uniform family of quantum circuits $\{Q_n: n \in \mathbb{N}\}$, such that
\begin{enumerate}
    \item For all $n \in \mathbb{N}$, $Q_n$ takes an $n$-qubit computational basis state as input, apply $Q_n$ on the input state, and measures the first qubit in the computational basis as output.
    \item For all $z \in L$, the probability that output of $Q_{|z|}$ applying on the input $z$ is one is greater than or equal to $2/3$.
    \item For all $z \notin L$, the probability that output of $Q_{|z|}$ applying on the input $z$ is zero is greater than or equal to $2/3$.
\end{enumerate}
If we have the randomized classical algorithm that can compute $f(x)$, then for all $z$: input bitstring, we consider the unitary quantum neural network given by
\begin{equation}
    U_{\text{QNN}} = Q_{|z|} \bigotimes_{i=1}^n X^{z_i}_i,
\end{equation}
where $X_i$ is the Pauli-X matrix acting on the $i$-th qubit, and the observable $O$ is given by $Z_1$.
Hence, we have
\begin{enumerate}
    \item For all $z \in L$, $f(x) = \bra{x} U_{\text{QNN}}^\dagger O U_{\text{QNN}} \ket{x} = \bra{z} Q_{|z|}^\dagger Z_1 Q_{|z|} \ket{z} = \mathrm{Pr}[$the output of $Q_{|z|}$ applying on the input $z$ is one $] - \mathrm{Pr}[$the probability that output of $Q_{|z|}$ applying on the input $z$ is zero$] \geq 2/3 - 1/3 = 1/3$.
    \item For all $z \notin L$, $f(x) = \bra{x} U_{\text{QNN}}^\dagger O U_{\text{QNN}} \ket{x} = \bra{z} Q_{|z|}^\dagger Z_1 Q_{|z|} \ket{z} = \mathrm{Pr}[$the output of $Q_{|z|}$ applying on the input $z$ is one $] - \mathrm{Pr}[$the probability that output of $Q_{|z|}$ applying on the input $z$ is zero$] \leq 1/3 - 2/3 = -1/3$.
\end{enumerate}
By assumption, we can use the randomized classical algorithm to compute an estimate $\hat{f}(x)$ such that $|\hat{f}(x) - f(x)| < 0.15$ with high probability over the randomness of the classical algorithm.
Therefore with high probability, $\hat{f}(x) > 0$ if $z \in L$ and $\hat{f}(x) < 0$ if $z \notin L$.
We can use the indication of whether $\hat{f}(x)$ is positive or negative to determine if $z \in L$ or $z \notin L$ with high probability over the randomness of the classical algorithm.
This implies that $L \in \text{BPP}$.

Together, the existence of the randomized classical algorithm implies that $\text{BQP} \subseteq \text{BPP}$. By definition, we have $\text{BPP} \subseteq \text{BQP}$, hence $\text{BPP} = \text{BQP}$.
\end{proof}

We will now give a classical machine learning algorithm that could learn $f(x)$ efficiently using few samples.
Recall that the data point is given by $\{x_i\}_{i=1}^N$, where $x_{i} \in \mathbb{R}^p$.
Now, we consider a classical ML model with the kernel function $k(x_i, x_j) = (\sum_{l=1}^p x_{il} x_{jl})^2$, which can be evaluated in time linear in the dimension $p$. Note that this definition of kernel is equivalent to the quantum kernel $\Tr(\rho(x_i) \rho(x_j)) = |\braket {x_i}{x_j}|^2$ for the encoding $\ket{x_i} = \sum_{k=1}^p x_{ik} \ket{k}$.
We will now use the theoretical framework we developed in the main text (the section on testing quantum advantage). In particular, we will use the prediction error of quantum kernel method given in Eq.~\ref{eq:QKprederr}.
It shows that for any observable $O$ and quantum neural network $U_{\text{QNN}}$, the prediction error after training from $N$ data points $\{(x_i, y_i = f(x_i))\}$ is given by
\begin{equation}
    \mathbb{E}_{x \in \mathcal{D}} |h(x) - f(x)| \leq c \sqrt{\frac{\min(d, \Tr(O^2))}{N}},
\end{equation}
where $d$ is the Hilbert space dimension of $\{\rho(x_i)\}_{i=1}^N$.
Because we have $\rho(x_i) = \ket{x_i}\! \bra{x_i}$ and $\ket{x_i} = \sum_{k=1}^p x_{ik} \ket{k}$, the dimension of the Hilbert space is upper bounded by $p^2$.
Therefore,
\begin{equation}
    \mathbb{E}_{x \in \mathcal{D}} |h(x) - f(x)| \leq c \sqrt{\frac{\min(d, \Tr(O^2))}{N}} \leq c \sqrt{\frac{p^2}{N}}.
\end{equation}
This is the result stated in the main text.
For more details about the machine learning models, the prediction error bound, and the proof for the prediction error bound of quantum kernel methods, see Appendix~\ref{sec:proofmain}~and~\ref{sec:QKmethodanalysis}.

\section{Complexity-theoretic argument for the power of data}
\label{sec:powerdatarigor}

\begin{figure*}[t!]
\centering
\includegraphics[width=\textwidth]{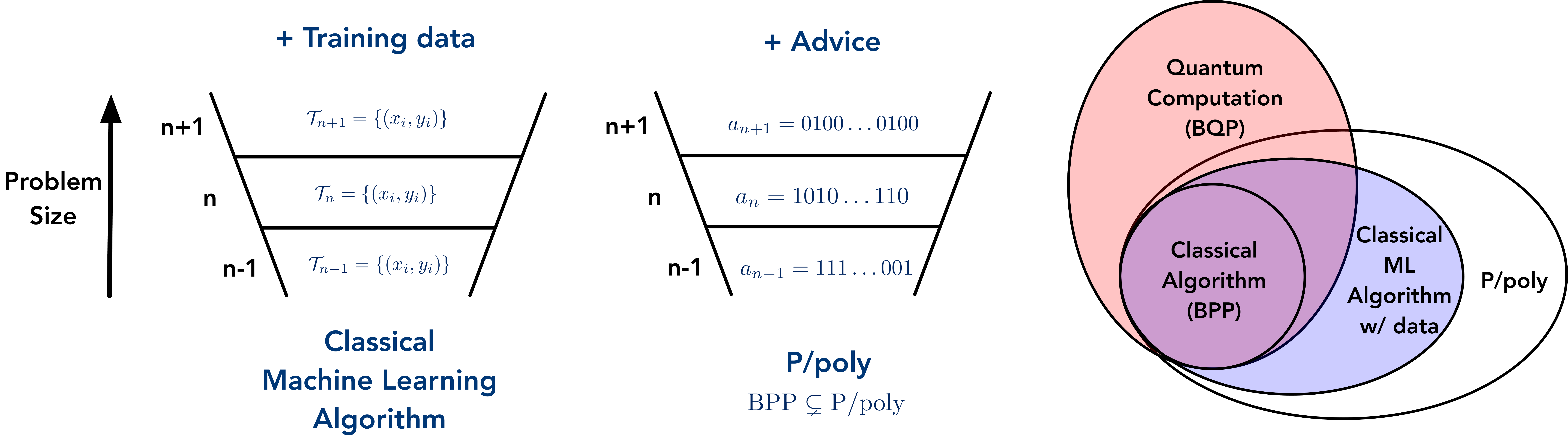}
    \caption{We present an illustration of the complexity class for classical machine learning algorithm with the availability of data. To the right, we have a diagram showing the relations between different complexity classes. 
    \label{fig:MLclass}}
\end{figure*}

In the main text, we give an argument based on an example to demonstrate the power of data. However, this is not satisfactory when we want to put the power of data on a rigorous footing.
To demonstrate this fact from a rigorous standpoint, let us capture classical ML algorithms that can learn from data by means of a complexity class, which we refer to as BPP/samp.
A language $L$ of bit strings is in BPP/samp if and only if the following holds:
There exists probabilistic Turing machines $D$ and $M$. $D$ generates samples $x$ with $|x| = n$ in polynomial time for any input size $n$. $D$ defines a sequence of input distributions $\{\mathcal{D}_n\}$. $M$ takes an input $x$ of size $n$ along with $\mathcal{T}=\left\{(x_i,y_i)\right\}_{i=1}^{\mathrm{poly}(n)}$ of polynomial size, where $x_i$ is sampled from $\mathcal{D}_n$ using Turing machine $D$ and $y_i$ conveys language membership: $y_i=1$ if $x_i\in L$ and $y_i=0$ if $x_i \not\in L$.
Moreover, we require
\begin{itemize}
    \item The probabilistic Turing machine $M$ to process all inputs $x$ in polynomial time (polynomial runtime).
    \item For all $x \in L$, $M$ outputs $1$ with probability greater than or equal to $2/3$ (prob.~completeness).
    \item For all $x \notin L$, $M$ outputs $1$ with probability less than or equal to $1/3$ (prob.~soundness).
\end{itemize}

If the Turing machine $M$ neglects the sampled data $\mathcal{T}$, this is equivalent to the definition of BPP.
Hence BPP is contained inside BPP/samp.

We can also see that $\mathcal{T}$ is a restricted form of randomized advice string.
It is not hard to show that BPP/samp is contained in P/poly based on the same proof strategy for Adleman's theorem.
We consider a new probabilistic Turing machine $M'$ that runs $M$ for $18n$ times. Each time, we use an independently sampled training set $\mathcal{T}$ from $\mathcal{D}_n$.
Then we take a majority vote from the $18n$ runs.
By Chernoff bound, the probability of failure for any given $x$ with $|x| = n$ would be at most $1/\mathrm{e}^n$. Hence by union bound, the probability that all $x$ with $|x| = n$ succeeds is at least $1 - (2 / \mathrm{e})^n$. This implies the existence of a particular choice of the $18n$ training sets and $18n$ random bit-strings used in each run of the probabilistic Turing machine $M$, such that for all $x$ with $|x| = n$ the decision of whether $x \in L$ is correct.
We simply define the advice string $a_n$ to one particular choice of the $18 n$ training sets and $18 n$ random bit-strings, which will be a string of size polynomial in $n$.
Hence we know that BPP/samp is contained in P/poly. An illustration is given in Figure \ref{fig:MLclass}.
We leave open the question of whether BPP/samp is strictly contained in P/poly.

The separation between P/poly and BPP is often illustrated by undecidable unary languages. The separation between BPP/samp and BPP could also be proved using a similar example.
Actually, an undecidable unary language serves as an equally good example.
Here, we choose to present a slightly more complicated example to demonstrate what BPP/samp could do.
Let us consider an undecidable unary language $L_{\mathrm{hard}} = \{1^n | n \in A\}$, where $A$ is a subset of the natural numbers $\mathbb{N}$ and a classically easy language $L_{\mathrm{easy}} \in \mathrm{BPP}$.
We assume that for every input size $n$, there exists an input $a_n \in L_{\mathrm{easy}}$ and an input $b_n \notin L_{\mathrm{easy}}$.
We define a new language as follows:
\begin{equation}
    L = \bigcup_{n=1}^\infty \{ x | \forall x \in L_{\mathrm{easy}}, 1^n \in L_{\mathrm{hard}}, |x| = n \} \cup \{ x | \forall x \notin L_{\mathrm{easy}}, 1^n \notin L_{\mathrm{hard}}, |x| = n \}.
\end{equation}
For each size $n$, if $1^n \in L_{\mathrm{hard}}$, $L$ would include all $x \in L_{\mathrm{easy}}$ with $|x| = n$. If $1^n \notin L_{\mathrm{hard}}$, $L$ would include all $x \notin L_{\mathrm{easy}}$ with $|x| = n$.
By definition, if we can output whether $x \in L$ for an input $x$ using a classical algorithm (BPP), we can output whether $1^n \in L_{\mathrm{hard}}$ by computing whether $x \in L_{\mathrm{easy}}$. This is however impossible due to the undecidability of $L_{\mathrm{hard}}$. Hence the language $L$ is not in BPP.
On the other hand, for every size $n$, a classical machine learning algorithm can use a single training data point $(x_0, y_0)$ to decide whether $x \in L$. An algorithm is as follows. Using $y_0$, we know whether $x_0 \in L_{\mathrm{easy}}$. Hence, we know whether $1^n \in L_{\mathrm{hard}}$. Then for any input $x$ with size $n$, we can output the correct answer by using the knowledge of whether $1^n \in L_{\mathrm{hard}}$ combined with a classical computation to decide whether $x \in L_{\mathrm{easy}}$.
This example nicely illustrates the power of data and how machine learning algorithms can utilize it. In summary, the data provide information that is hard to compute with a classical computer (e.g., whether $1^n \in L_{\mathrm{hard}}$). Then the classical machine learning algorithm would perform classical computation to infer the solution from the given knowledge (e.g., computing whether $x \in L_{\mathrm{easy}}$).
The same language $L$ also yields a separation between BPP/samp and BQP because $L$ is constructed to be undecidable.

From a practical perspective, it is impossible to obtain training data that is undecidable. But it is still possible to obtain data that cannot be efficiently computed with a classical computer, since the universe operates quantum mechanically. If the universe computes classically, then the data we can obtain will be computable by BPP and there is no separation between classical ML algorithm with data from BPP and BPP.
We now present a simple argument for a separation between classical algorithm learning with data coming from quantum computation and BPP. This follows from a similar argument as the previous example. Here, we assume that there is a sequence of quantum circuits such that the Z measurement on the first qubit (being $+1$ with probability $> 2/3$ or $< 1/3$) is hard to decide classically. This defines a unary language $L_{\mathrm{hard}}'$ that is outside BPP, but inside BQP.
We can then use $L_{\mathrm{hard}}'$ in replace of $L_{\mathrm{hard}}$ for the example above.
When the data comes from BQP, the class classical ML algorithms that can learn from the data would not have a separation from BQP.

\section{Relation between quantum kernel methods and quantum neural networks}
\label{sec:qk_qnn}
In this section we demonstrate the formal equivalence of an arbitrary depth neural network with a quantum kernel method built from the original quadratic quantum kernel.  This connection helps demonstrate the feature map induced by this kernel to motivate its use as opposed to the simpler inner product.  While this equivalence shows the flexibility of this quantum kernel, it does not imply that it allows learning with a parsimonious amount of data.  Indeed, in many cases it requires both an exponential amount of data and exponential precision in evaluation due to the fidelity type metric.  In later sections we show simple cases where it fails for illustration purposes.

\begin{proposition}
Training an arbitrarily deep quantum neural network $U_{\mathrm{QNN}}$ with a trainable observable $O$
is equivalent to training a quantum kernel method with kernel $k_Q(x_i, x_j) = \Tr(\rho(x_i) \rho(x_j))$.
\end{proposition}
\begin{proof}
Let us define $\rho_i = \rho(x_i) = U_{\mathrm{enc}}(x_i) \ket{0^n}\! \bra{0^n} U_{\mathrm{enc}}(x_i)^\dagger$ to be the corresponding quantum states for the input vector $x_i$.
The training of a quantum neural network can be written` as
\begin{equation}
    \min_{U \in \mathcal{C} \subset U(2^n)} \sum_{i=1}^N l(\Tr(O U \rho_i U^\dagger), y_i),
\end{equation}
where $l(\tilde{y}, y)$ is a loss function that measures how close the prediction $\tilde{y}$ is to the true label $y$, $\mathcal{C}$ is the space of all possible unitaries considered by the parameterized quantum circuit, $O$ is some predefined observable that we measure after evolving with $U$.
Let us denote the optimal $U$ to be $U^*$, then the prediction for a new input $x$ is given by
$\Tr(O U^* \rho(x) (U^*)^\dagger)$.

On the other hand, the training of the quantum kernel method under the implied feature map is equivalent to training $W \in  \mathbb{C}^{2^n \times 2^n}$ under the optimization
\begin{equation}
    \min_{W \in \mathbb{C}^{2^n \times 2^n}} \sum_{i=1}^N l(\Tr(W \rho_i), y_i) + \lambda \Tr(W^\dagger W),
\end{equation}
where $\lambda \geq 0$ is the regularization parameter and $l(\tilde{y}, y)$ is the loss function.
Let us denote the optimal $W$ to be $W^*$, then the prediction for a new input $x$ is given by
$\Tr(W^* \rho(x))$.
The well-known kernel trick allows efficient implementation of this machine learning model, and connects the original quantum kernel to the derivation here.
Using the fact that $\rho_i$ is Hermitian and set $\lambda = 0$,
the quantum kernel method can be expressed as
\begin{equation}
    \min_{\substack{U \in U(2^n), \\ O \in \mathbb{C}^{2^n \times 2^n}, O = O^\dagger}} \sum_{i=1}^N l(\Tr(O U \rho_i U^\dagger), y_i).
\end{equation}
This is equivalent to training an arbitrarily deep quantum neural network $U$ with a trainable observable $O$.
\end{proof}

\section{Proof of a general form of prediction error bound}
\label{sec:proofmain}
 
This section is dedicated to deriving the precise statement for the core prediction error bound from which we base our methodology: $\mathbb{E}_x|h(x) - f(x)| \leq \mathcal{O}(\sqrt{s / N})$ given by the first inequality in Equation~\eqref{eq:main_pred}. We will provide a detailed proof for the following general theorem when we include the regularization parameter $\lambda$. The regularization parameter $\lambda$ will be used to improve prediction performance by limiting the complexity of the machine learning model.

\begin{theorem} \label{thm:main}
Consider an observable $O$ with $\norm{O}_\infty \leq 1$, a quantum unitary $U$ (e.g., a quantum neural network or a general Hamiltonian evolution), a mapping of classical vector $x$ to quantum system $\rho(x)$, and a training set of $N$ data $\{(x_i, y_i = \Tr(O^U \rho(x_i)) )\}_{i=1}^N$, with $O^U = U^\dagger O U$ being the Heisenberg evolved observable.
The training set is sampled from some unknown distribution over the input $x$.
Suppose that $k(x, x')$ can be evaluated efficiently and the kernel function is re-scaled to satisfy $\sum_{i=1}^N k(x_i, x_i) = N$. Define the Gram matrix $K_{ij} = k(x_i, x_j)$. For any $\lambda \geq 0$, with probability at least $1-\delta$ over the sampling of the training data, we can learn a model $h(x)$ from the training data, such that the expected prediction error is bounded by
\begin{equation} \label{eq:mainprederr}
    \mathbb{E}_{x} |h(x) - \Tr(O^U \rho(x))| \leq \mathcal{O}\left(\sqrt{ \frac{\Tr(A_{\mathrm{tra}} O^U \otimes O^U)}{N}} + \sqrt{\frac{\Tr(A_{\mathrm{gen}} O^U \otimes O^U)}{N}} + \sqrt{\frac{\log(1 / \delta)}{N}}\right),
\end{equation}
where the two operators $A_{\mathrm{tra}}, A_{\mathrm{gen}}$ are given as
\begin{align}
A_{\mathrm{tra}} & = \lambda^2 \sum_{i=1}^N \sum_{j=1}^N ((K + \lambda I)^{-2})_{ij} \rho(x_i)\otimes \rho(x_j), \label{eq:Atra} \\
A_{\mathrm{gen}} & = \sum_{i=1}^N \sum_{j=1}^N ((K + \lambda I)^{-1} K (K + \lambda I)^{-1})_{ij} \rho(x_i)\otimes \rho(x_j).    \label{eq:Agen}
\end{align}
This is a data-dependent bound as $A_{\mathrm{tra}}$ and $A_{\mathrm{gen}}$ both depend on the $N$ training data.
\end{theorem}

When we take the limit of $\lambda \rightarrow 0$, we have $A_{\mathrm{tra}} = 0$ and $A_{\mathrm{gen}} = \sum_{i=1}^N \sum_{j=1}^N (K^{-1})_{ij} \rho(x_i)\otimes \rho(x_j)$.
Thus with probability at least $0.99 = 1 - \delta$, we have
\begin{equation} \label{eq:mainprederr_noreg}
    \mathbb{E}_{x} |h(x) - \Tr(O^U \rho(x))| \leq \mathcal{O}\left( \sqrt{\frac{s_K(N)}{N}}\right),
\end{equation}
where $s_K(N) = \sum_{i=1}^N \sum_{j=1}^N (K^{-1})_{ij} \Tr(O^U \rho(x_i)) \Tr(O^U \rho(x_j))$. This is the formula stated in the main text.
However, in practice, we would recommend the use of regularization $\lambda > 0$ to prevent numerical instability and to obtain prediction error bound when we use a regularized ML model.
 
In Section~\ref{sec:MLmodels}, we will present the definition of the machine learning models used to prove Theorem~\ref{thm:main}. In Section~\ref{sec:trainingerror}~and~\ref{sec:generalizationerror}, we will analyze the training error and generalization error of the machine learning models we consider to prove the prediction error bound given in Theorem~\ref{thm:main}.

\subsection{Definition and training of machine learning models} 
\label{sec:MLmodels}
 
We consider a class of machine learning models, including Gaussian kernel regression, infinite-width neural networks, and quantum kernel methods.
These models are equivalent to training a linear function mapping from a (possibly infinite-dimensional) Hilbert space $\mathcal{H}$ to $\mathbb{R}$.
The linear function can be written as $\langle w, \phi(x) \rangle$, where $w$ parameterizes the linear function, $\langle \cdot, \cdot \rangle: \mathcal{H}\times \mathcal{H} \rightarrow \mathbb{R}$ is an inner product, and $\phi(x)$ is a nonlinear mapping from the input vector $x$ to the Hilbert space $\mathcal{H}$.
For example, in quantum kernel method, we use the space of $2^n \times 2^n$ Hermitian matrices as the Hilbert space $\mathcal{H}$. This yields a natural definition of inner product $\langle \rho, \sigma \rangle = \Tr(\rho \sigma) \in \mathbb{R}$.

Because the output $y = \Tr(U^\dagger O U \rho(x))$ of the quantum model satisfies $y \in [-1, 1]$, we confine the output of the machine learning model to the interval $[-1, 1]$. The resulting machine learning model would be
\begin{equation}
    h_w(x) = \min(1, \max(-1, \langle w, \phi(x) \rangle)).
\end{equation}
For efficient optimization of $w$, we consider minimization of the following loss function
\begin{equation} \label{eq:optw}
    \min_{w} \lambda \langle w, w \rangle + \sum_{i=1}^N \left( \langle w, \phi(x) \rangle - \Tr(U^\dagger O U \rho(x_i)) \right)^2,
\end{equation}
where $\lambda \geq 0$ is a hyper-parameter.
We define $\Phi = (\phi(x_1), \ldots, \phi(x_N))$.
The kernel matrix $K = \Phi^\dagger \Phi$ is an $N \times N$ matrix that defines the geometry between all pairs of the training data.
We see that $K_{ij} = \langle \phi(x_i), \phi(x_j) \rangle = k(x_i, x_j) \in \mathbb{R}$.
Without loss of generality, we consider $\Tr(K) = N$, which can be done by rescaling $k(x_i, x_j)$.
The optimal $w$ can be written down explicitly as
\begin{equation} \label{eq:optw-soln}
    w = \sum_{i=1}^N \sum_{j=1}^N \phi(x_i) ((K+\lambda I)^{-1})_{ij} \Tr(U^\dagger O U \rho(x_j)).
\end{equation}
Hence the trained machine learning model would be
\begin{equation}
    h_w(x) = \min\left(1, \max\left(-1, \sum_{i=1}^N \sum_{j=1}^N k(x_i, x) ((K+\lambda I)^{-1})_{ij} \Tr(U^\dagger O U \rho(x_j))\right)\right).
\end{equation}
This is an analytic representation for various trained machine learning models, including least-square support vector machine \cite{suykens1999least}, kernel regression \cite{nadaraya1964estimating, altman1992introduction}, and infinite-width neural networks \cite{jacot2018neural}.
We will now analyze the prediction error of these machine learning models:
\begin{equation} \label{eq:errorw}
\epsilon_w(x) = |h_w(x) - \Tr(U^\dagger O U \rho(x))|,
\end{equation}
which is uniquely determined by the kernel matrix $K$ and the hyper-parameter $\lambda$. In particular, we will focus on providing an upper bound on the expected prediction error
\begin{align}
    \mathbb{E}_{x} \,\, \epsilon_w(x) = \underbrace{\frac{1}{N}\sum_{i=1}^N \epsilon_w(x_i)}_{\text{Training error}} + \underbrace{\mathbb{E}_{x} \,\, \epsilon_w(x) - \frac{1}{N}\sum_{i=1}^N \epsilon_w(x_i)}_{\text{Generalization error}},
\end{align}
which is the sum of training error and generalization error.

\subsection{Training error}
\label{sec:trainingerror}

We will now relate the training error to the optimization problem, i.e., Equation~\eqref{eq:optw}, for obtaining the machine learning model $h_w(x)$.
Because $\norm{O} \leq 1$, we have $\Tr(U^\dagger O U \rho(x)) \in [-1, 1]$, and hence
$\epsilon_w(x) = |h_w(x) - \Tr(U^\dagger O U \rho(x))| \leq |\langle w, \phi(x) \rangle - \Tr(U^\dagger O U \rho(x))|$.
Using the convexity of $x^2$ and Jensen's inequality, we obtain
\begin{equation}
    \frac{1}{N}\sum_{i=1}^N \epsilon_w(x_i) \leq \sqrt{\frac{1}{N} \sum_{i=1}^N \left( \langle w, \phi(x) \rangle - \Tr(U^\dagger O U \rho(x_i)) \right)^2}.
\end{equation}
We can plug in the expression for the optimal $w$ given in Equation~\eqref{eq:optw-soln} to yield
\begin{equation}
    \frac{1}{N}\sum_{i=1}^N \epsilon_w(x_i) \leq \sqrt{ \frac{\Tr(A_{\mathrm{tra}} (U^\dagger O U) \otimes (U^\dagger O U))}{N} },
\end{equation}
where $A_{\mathrm{tra}} = \lambda^2 \sum_{i=1}^N \sum_{j=1}^N ((K + \lambda I)^{-2})_{ij} \rho(x_i)\otimes \rho(x_j)$.
When $K$ is invertible and $\lambda = 0$, we can see that the training error is zero. However, in practice, we often set $\lambda > 0$.

\subsection{Generalization error}
\label{sec:generalizationerror}

A basic theorem in statistics and learning theory is presented below. This theorem provides an upper bound on the largest (one-sided) deviation from expectation over a family of functions.

\begin{theorem}[See Theorem 3.3 in \cite{mohri2018foundations}] \label{thm:rademacher}
Let $\mathcal{G}$ be a family of function mappings from a set $\mathcal{Z}$ to $[0, 1]$.
Then for any $\delta > 0$, with probability at least $1 - \delta$ over identical and independent draw of $N$ samples from $\mathcal{Z}$: $z_1, \ldots, z_N$, we have for all $g \in \mathcal{G}$,
\begin{equation}
    \mathbb{E}_z[g(z)] \leq \frac{1}{N} \sum_{i=1}^N g(z_i) + 2 \mathbb{E}_\sigma \left[ \sup_{g \in \mathcal{G}} \frac{1}{N} \sum_{i=1}^N \sigma_i g(z_i) \right] + 3 \sqrt{\frac{\log(2 / \delta)}{2 N}},
\end{equation}
where $\sigma_1, \ldots \sigma_N$ are independent and uniform random variables over ${\pm 1}$.
\end{theorem}

For our purpose, we will consider $\mathcal{Z}$ to be the space of input vector with $z_i = x_i$ drawn from some input distribution.
Each function $g$ would be equal to $\epsilon_w / 2$ for some $w$, where $\epsilon_w$ is defined in Equation~\eqref{eq:errorw}.
The reason that we divide by $2$ is because the range of $\epsilon_w$ is $[0, 2]$.
And $\forall \gamma = 1, 2, 3, \ldots$, we define $\mathcal{G}_\gamma$ to be $\{\epsilon_w / 2 \,\, | \,\, \forall \norm{w} \leq \gamma\}$.
The definition of an infinite sequence of family of functions $\mathcal{G}_\gamma$ is useful for proving a prediction error bound for an unbounded class of machine learning models $h_w(x)$, where $\norm{w}$ could be arbitrarily large.
Using Theorem~\ref{thm:rademacher} and multiplying the entire inequality by $2$, we can show that the following inequality holds for any $w$ with $\norm{w} \leq \gamma$,
\begin{equation} \label{ineq:RMC}
    \mathbb{E}_{x}[\epsilon_w(x)] - \frac{1}{N} \sum_{i=1}^N \epsilon_w(x_i) \leq 2 \mathbb{E}_\sigma \left[ \sup_{\norm{v} \leq \gamma} \frac{1}{N} \sum_{i=1}^N \sigma_i \epsilon_v(x_i) \right] + 6 \sqrt{\frac{\log(4 \gamma^2 / \delta)}{2 N}},
\end{equation}
with probability at least $1 - \delta / 2 \gamma^2$. This probabilistic statement holds for any $\gamma = 1, 2, 3, \ldots$, but this does not yet guarantee that the inequality holds for all $\gamma$ with high probability.
We need to apply a union bound over all $\gamma$ to achieve this, which shows that Inequality~\eqref{ineq:RMC} holds for all $\gamma$ with probability at least $1 - \sum_{\gamma=1}^\infty \delta / 2 \gamma^2 \geq 1 - \delta$.

Together we have shown that, for any $w \in \mathcal{H}$, the generalization error $\mathbb{E}_{x}[\epsilon_w(x)] - \frac{1}{N} \sum_{i=1}^N \epsilon_w(x_i)$ is upper bounded by
\begin{equation} \label{eq:radmacall}
    2 \mathbb{E}_\sigma \left[ \sup_{\norm{v} \leq \lceil \norm{w} \rceil} \frac{1}{N} \sum_{i=1}^N \sigma_i \epsilon_v(x_i) \right] + 6 \sqrt{\frac{\log(4 \lceil \norm{w} \rceil^2 / \delta)}{2 N}},
\end{equation}
with probability at least $1 - \delta$, where we consider the particular inequality with $\gamma = \lceil \norm{w} \rceil$. We will now analyze the above inequality using Talagrand's contraction lemma.

\begin{lemma}[Talagrand's contraction lemma; See Lemma 5.7 in \cite{mohri2018foundations}]
Let $\mathcal{G}$ be a family of function from a set $\mathcal{Z}$ to $\mathbb{R}$.
Let $l_1, \ldots, l_N$ be Lipschitz-continuous function from $\mathbb{R} \rightarrow \mathbb{R}$ with Lipschitz constant $L$. Then 
\begin{equation}
    \mathbb{E}_\sigma \left[ \sup_{g \in \mathcal{G}} \frac{1}{N} \sum_{i=1}^N \sigma_i l_i(g(z_i)) \right] \leq L \mathbb{E}_\sigma \left[ \sup_{g \in \mathcal{G}} \frac{1}{N} \sum_{i=1}^N \sigma_i g(z_i) \right].
\end{equation}
\end{lemma}

We consider $l_i(s) = |\min(1, \max(-1, s)) - \Tr(U^\dagger O U \rho(x_i))|$, $z_i = x_i$, and $\mathcal{G} = \{g_v(z_i) = \langle v, z_i \rangle \, | \, \norm{v} \leq \lceil \norm{w} \rceil \}$.
This choice of functions gives $\epsilon_v(x_i) = l_i(g(z_i))$. Furthermore, $l_i$ is Lipschitz-continuous with Lipschitz constant~$1$. Talagrand's contraction lemma then allows us to bound the formula in Equation~\eqref{eq:radmacall} by
\begin{align}
    & 2 \mathbb{E}_\sigma \left[ \sup_{\norm{v} \leq \lceil \norm{w} \rceil} \frac{1}{N} \sum_{i=1}^N \sigma_i \langle v, \phi(x_i) \rangle \right] + 6 \sqrt{\frac{\log(4 \lceil \norm{w} \rceil^2 / \delta)}{2 N}} \\
    & \leq 2 \mathbb{E}_\sigma \left[ \sup_{\norm{v} \leq \lceil \norm{w} \rceil} \frac{1}{N} \norm{v} \norm{\sum_{i=1}^N \sigma_i \phi(x_i)} \right] + 6 \sqrt{\frac{\log(4 \lceil \norm{w} \rceil^2 / \delta)}{2 N}} \\
    & \leq 2 \lceil \norm{w} \rceil \mathbb{E}_\sigma \left[ \frac{1}{N} \norm{\sum_{i=1}^N \sigma_i \phi(x_i)} \right] + 6 \sqrt{\frac{\log(4 \lceil \norm{w} \rceil^2 / \delta)}{2 N}} \\
    & \leq 2 \frac{\lceil \norm{w} \rceil}{N} \sqrt{\mathbb{E}_\sigma \sum_{i=1}^N \sum_{j=1}^N \sigma_i \sigma_j k(x_i, x_j)} + 6 \sqrt{\frac{\log(4 \lceil \norm{w} \rceil^2 / \delta)}{2 N}}\\
    & \leq 2 \frac{\sqrt{ \lceil \norm{w} \rceil^2 \Tr(K)}}{N} + 6 \sqrt{\frac{\log(4 \lceil \norm{w} \rceil^2 / \delta)}{2 N}} \\
    & \leq 2 \sqrt{\frac{ \lceil \norm{w} \rceil^2 }{N}} + 6 \sqrt{\frac{\log(\lceil \norm{w} \rceil)}{N}} + 6 \sqrt{\frac{\log(4 / \delta)}{2 N}} \\
    & \leq 8 \sqrt{\frac{ \lceil \norm{w} \rceil^2 }{N}} + 6 \sqrt{\frac{\log(4 / \delta)}{2 N}}.
\end{align}
The first inequality uses Cauchy's inequality. The second inequality uses the fact that $\norm{v} \leq \lceil \norm{w} \rceil$. The third inequality uses a Jensen's inequality to move $\mathbb{E}_\sigma$ into the square-root. The fourth inequality uses the fact that $\sigma_i$ are independent and uniform random variable taking $+1, -1$. The fifth inequality uses $\sqrt{x + y} \leq \sqrt{x} + \sqrt{y}, \forall x, y \geq 0$ and our assumption that we rescale $K$ such that $\Tr(K) = N$. The sixth inequality uses the fact that $x^2 \geq \log(x), \forall x \in \mathbb{N}$.

Finally, we plug in the optimal $w$ given in Equation~\eqref{eq:optw-soln}. This allows us to obtain an upper bound of the generalization error:
\begin{equation}
    \mathbb{E}_{x}[\epsilon_w(x)] - \frac{1}{N} \sum_{i=1}^N \epsilon_w(x_i) \leq 8 \frac{ \lceil \sqrt{\Tr(A_{\mathrm{gen}} (U^\dagger O U) \otimes (U^\dagger O U))} \rceil}{\sqrt{N}} + 6 \sqrt{\frac{\log(4 / \delta)}{2 N}},
\end{equation}
where $A_{\mathrm{gen}} = \sum_{i=1}^N \sum_{j=1}^N ((K + \lambda I)^{-1} K (K + \lambda I)^{-1})_{ij} \rho(x_i)\otimes \rho(x_j)$. When $K$ is invertible and $\lambda = 0$, we have $A_{\mathrm{gen}} = \sum_{i=1}^N \sum_{j=1}^N (K^{-1})_{ij} \rho(x_i)\otimes \rho(x_j)$.

\section{Simplified prediction error bound based on dimension and geometric difference}
\label{sec:prooftwo}

In this section, we will show that for quantum kernel methods, we have
\begin{equation} \label{eq:predqkm}
    \mathbb{E}_{x} |h^{\mathrm{Q}}(x) - \Tr(O^U \rho(x))| \leq \mathcal{O}\left( \sqrt{\frac{\min(d, \Tr(O^2))}{N}}\right),
\end{equation}
where $d$ is the dimension of the training set space $d =  \mathrm{dim}(\mathrm{span}(\rho(x_1), \ldots, \rho(x_N)))$.
If we use the quantum kernel method as a reference point, then the prediction error of another machine learning algorithm that produces $h(x)$ using kernel matrix $K$ can be bounded by
\begin{equation}
    \mathbb{E}_{x} |h(x) - \Tr(O^U \rho(x))| \leq \mathcal{O}\left( g \sqrt{\frac{\min(d, \Tr(O^2))}{N}}\right),
\end{equation}
where $g = \sqrt{\norm{\sqrt{K^
\mathrm{Q}} K^{-1} \sqrt{K^\mathrm{Q}}}_\infty}$ assuming the normalization condition $\Tr(K^\mathrm{Q}) = \Tr(K) = N$.

\subsection{Quantum kernel method} \label{sec:QKmethodanalysis}

In quantum kernel method, the kernel function that will be used to train the model is defined using the quantum Hilbert space $k_Q(x, x') = \Tr(\rho(x) \rho(x'))$. Correspondingly, we define the kernel matrix $K^{Q}_{ij} = k_{Q}(x_i, x_j)$. We will focus on $\rho(x)$ being a pure state, so the scaling condition $\Tr(K^{\mathrm{Q}}) = \sum_{i=1}^N k_Q(x_i, x_i) = N$ is immediately satisfied. We also denote the trained model as $h^{\mathrm{Q}}$ for the quantum kernel method.
We now consider an orthonormal basis $\{\sigma_1, \ldots, \sigma_d\}$ for the $d$-dimensional quantum state space formed by the training data $\mathrm{span}\{\rho(x_1), \ldots, \rho(x_N)\}$ under the inner product $\langle \rho, \sigma \rangle = \Tr(\rho \sigma)$. We have $\sigma_p$ is Hermitian, $\Tr(\sigma_p^2) = 1$, but $\sigma_p$ may not be positive semi-definite.

We consider an expansion of $\rho(x_i)$ in terms of $\sigma_p$:
\begin{equation} \label{eq:rhosigma}
    \rho(x_i) = \sum_{p=1}^d \alpha_{i p} \sigma_p,
\end{equation}
where $\alpha \in \mathbb{R}^{N \times d}$. The coefficient $\alpha$ is real as the vector space of Hermitian matrices is over real numbers. Note that multiplying a Hermitian matrix with an imaginary number will not generally result in a Hermitian matrix, hence Hermitian matrices are not a vector space over complex numbers.
We can perform a singular value decomposition on $\alpha = U \Sigma V^\dagger$, where $U \in \mathbb{C}^{N \times d}, \Sigma, V \in \mathbb{C}^{d \times d}$ with $U^\dagger U = I$, $\Sigma$ is diagonal and $\Sigma \succ 0$, $V^\dagger V = V V^\dagger = I$.
Then $K^{\mathrm{Q}} = \alpha \alpha^\dagger = U \Sigma^2 U^\dagger$.
This allows us to explicitly evaluate $A_{\mathrm{tra}}$ and $A_{\mathrm{gen}}$ given in Equation~\eqref{eq:Atra}~and~\eqref{eq:Agen}:
\begin{align}
A_{\mathrm{tra}} & = \lambda^2 \sum_{p=1}^d \sum_{q=1}^d \left(V \frac{\Sigma^2}{(\Sigma^2 + \lambda I)^2} V^\dagger\right)_{pq} \sigma_p \otimes \sigma_q, \\
A_{\mathrm{gen}} & = \sum_{p=1}^d \sum_{q=1}^d \left(V \frac{\Sigma^4}{(\Sigma^2 + \lambda I)^2} V^\dagger\right)_{pq} \sigma_p \otimes \sigma_q,
\end{align}
which can be done by expanding $\rho(x_i)$ in terms of $\sigma_p$.
Because $\Sigma \succ 0$, when we take the limit of $\lambda \rightarrow 0$, we have $A_{\mathrm{tra}} = 0$ and $A_{\mathrm{gen}} = \sum_{p=1}^d \sum_{q=1}^d \delta_{pq} \sigma_p \otimes \sigma_q = \sum_{p=1}^d \sigma_p \otimes \sigma_p$. Hence $\Tr(A_{\mathrm{tra}} O^U \otimes O^U) = 0$ and $\Tr(A_{\mathrm{gen}} O^U \otimes O^U) = \sum_{p=1}^d \Tr(\sigma_p O^U)^2$. From Equation~\eqref{eq:mainprederr} with $\lambda \rightarrow 0$, we have
\begin{equation} \label{eq:prederrQK}
    \mathbb{E}_{x} |h^{Q}(x) - \Tr(O^U \rho(x))| \leq \mathcal{O}\left(\sqrt{\frac{\sum_{p=1}^d \Tr( O^U \sigma_p)^2}{N}} + \sqrt{\frac{\log(1 / \delta)}{N}}\right).
\end{equation}
Because $\{\sigma_1, \ldots, \sigma_k\}$ forms an orthonormal set in the space of $2^n \times 2^n$ Hermitian matrices, $\sum_{p=1}^d \Tr( O^U \sigma_p)^2$ is the Frobenius norm of the observable $O^U$ restricted to the subspace $\mathrm{span}\{\sigma_1, \ldots, \sigma_k\}$.

We now focus on obtaining an informative upper bound on how large $\sum_{p=1}^d \Tr( O^U \sigma_p)^2$ could be. First, because we can extend the subspace $\mathrm{span}\{\sigma_1, \ldots, \sigma_k\}$ to the full Hilbert space $\mathrm{span}\{\sigma_1, \ldots \sigma_{4^n}\}$, we have $\sum_{p=1}^d \Tr( O^U \sigma_p)^2 \leq \sum_{p=1}^{4^n} \Tr( O^U \sigma_p)^2 =  \Tr((O^U)^2) = \norm{O^U}_F^2$.
Next, we will show that $\sum_{p=1}^d \Tr( O^U \sigma_p)^2 \leq d \norm{O^U}_\infty^2 \leq d$, where $\norm{O^U}_\infty$ is the spectral norm of the observable $O^U$. We pick a linearly-independent set of $\{\rho_1, \ldots, \rho_k\}$ from $\{\rho(x_1), \ldots \rho(x_N)\}$. We assume that all the quantum states are pure, hence we have $\rho_i = \ket{\psi_i}\!\bra{\psi_i}, \forall i = 1, \ldots, d$.
The pure states $\{\ket{\psi_1}, \ldots, \ket{\psi_k}\}$ may not be orthogonal, so we perform a Gram-Schmidt process to create an orthonormal set of quantum states $\{\ket{\phi_1}, \ldots, \ket{\phi_k}\}$.
Because $\rho_i$ are linear combination of $\ket{\phi_q} \bra{\phi_r}, \forall q, r = 1, \ldots, d$, we have
\begin{equation}
    \sigma_p = \sum_{q=1}^d \sum_{r=1}^d s_{pqr} \ket{\phi_q}\!\bra{\phi_r}, \forall p = 1, \ldots, d.
\end{equation}
The condition $\Tr(\sigma_p \sigma_{p'}) = \delta_{p p'}$ implies that $\sum_{q=1}^d \sum_{r=1}^d s_{pqr} s_{p' qr} = \delta_{p p'}$.
If we view $s$ as a vector $\vec{s}$ of size $d^2$, then $\langle \vec{s}_p, \vec{s}_{p'} \rangle = \delta_{p p'}$. Thus $\{\vec{s}_1, \ldots, \vec{s}_k\}$ forms a set of orthonormal vectors in $\mathbb{R}^{d^2}$, which implies $\sum_{p=1}^d \vec{s}_p \vec{s}_p^\dagger \preceq I$.
Let us define the projection operator $P = \sum_{q=1}^d \ket{\phi_q}\!\bra{\phi_q}$.
We will also consider a vector $\vec{o} \in \mathbb{R}^{d^2}$, where $\vec{o}_{qr} = \bra{\phi_r} O^U \ket{\phi_q}$.
We have 
\begin{align}
    \sum_{p=1}^d \Tr( O^U \sigma_p)^2 & = \sum_{p=1}^d \left( \sum_{q=1}^d \sum_{r=1}^d s_{pqr} \bra{\phi_r} O^U \ket{\phi_q} \right)^2 = \sum_{p=1}^d \vec{o}^\dagger \vec{s}_p \vec{s}_p^\dagger \vec{o} \\
     &  \leq \vec{o}^\dagger \vec{o} = \left( \sum_{q=1}^d \sum_{r=1}^d \bra{\phi_r} O^U \ket{\phi_q} \right)^2 = \norm{P O^U P}_F^2.
\end{align}
The inequality comes from the fact that $\sum_{p=1}^d \vec{s}_p \vec{s}_p^\dagger \preceq I$.
With a proper choice of basis, one could view $P O^U P$ as an $d \times d$ matrix.
Hence $\norm{P O^U P}_F \leq \sqrt{d} \norm{P O^U P}_\infty \leq \sqrt{d} \norm{O^U}_\infty$.
This established the fact that $\sum_{p=1}^d \Tr( O^U \sigma_p)^2 \leq d \norm{O^U}_\infty^2 \leq d$.
Combining with Equation~\eqref{eq:prederrQK}, we have
\begin{equation}
    \mathbb{E}_{x} |h^{Q}(x) - \Tr(O^U \rho(x))| \leq \mathcal{O}\left(\sqrt{\frac{\min(d, \norm{O^U}_F^2)}{N}} + \sqrt{\frac{\log(1 / \delta)}{N}}\right).
\end{equation}
This elucidates the fact that the prediction error of a quantum kernel method is bounded by minimum of the dimension of the quantum subspace formed by the training set and the Frobenius norm of the observable $O^U$.

Choosing a small but non-zero $\lambda$ allows us to consider an approximate space of $\mathrm{span}\{\rho(x_1), \ldots, \rho(x_N)\}$ formed by the training set. The training error $\sqrt{\Tr(A_{\mathrm{tra}} O^U \otimes O^U) / N}$ would increase slightly, and the generalization error $\sqrt{\Tr(A_{\mathrm{gen}} O^U \otimes O^U) / N}$ would reflect the Frobenius norm of $O^U$ restricted to a smaller subspace, which only contains the principal components of the space formed by the training set. This would be a better choice when most states lie in low-dimensional subspace with small random fluctuations.
One may also consider training a machine learning model with truncated kernel matrix $K_\lambda$, where all singular values below $\lambda$ are truncated. This makes the act of restricting to an approximate subspace more explicit.

\subsection{Another machine learning method compared to quantum kernel}
\label{sec:MLother}

We now consider an upper bound on the prediction error using the quantum kernel method as a reference point for some machine learning algorithm.
For the following discussion, we consider classical neural networks with large hidden sizes.
The function generated by a classical neural network with large hidden size after training is equivalent to the function $h(x)$ given in Equation~\eqref{eq:optw-soln} with $\lambda = 0$ and with a special kernel function $k_{\mathrm{NTK}}(x, x')$ known as the neural tangent kernel (NTK) \cite{jacot2018neural}.
The precise definition of $k_{\mathrm{NTK}}(x, x')$ depends on the architecture of the neural network. For example, a two-layer feedforward neural network (FNN), a three-layer FNN, or some particular form of convolutional neural network (CNN) all correspond to different $k_{\mathrm{NTK}}(x, x')$.
Given the kernel $k_{\mathrm{NTK}}(x, x')$, we can define the kernel matrix $\tilde{K}_{ij} = k_{\mathrm{NTK}}(x_i, x_j)$.
For neural tangent kernel, the scaling condition $\Tr(\tilde{K}) = \sum_{i=1}^N k_{\mathrm{NTK}}(x_i, x_i) = N$ may not be satisfied. Hence, we define a normalized kernel matrix $K = N \tilde{K} / \Tr(\tilde{K})$.
When $\lambda = 0$, the trained machine learning model (given in Equation~\eqref{eq:optw-soln}) under the normalized matrix $K$ and the original matrix $\tilde{K}$ are the same.
In order to apply Theorem~\ref{thm:main}, we will use the normalized kernel matrix $K$ for the following discussion.
From Equation~\eqref{eq:mainprederr} with $\lambda = 0$, we have
\begin{equation} \label{eq:NNprederr}
    \mathbb{E}_{x} |h(x) - \Tr(O^U \rho(x))| \leq \mathcal{O}\left(\sqrt{\frac{\Tr(A O^U \otimes O^U)}{N}} + \sqrt{\frac{\log(1 / \delta)}{N}}\right),
\end{equation}
where $A = \sum_{i=1}^N \sum_{j=1}^N (K^{-1})_{ij} \rho(x_i)\otimes \rho(x_j)$.
Using Equation~\eqref{eq:rhosigma} on the expansion of $\rho(x_i)$, we have
\begin{align}
    A & = \sum_{p=1}^d \sum_{q=1}^d \sum_{i=1}^N \sum_{j=1}^N (K^{-1})_{ij} \alpha_{ip} \alpha_{jq} \sigma_p \otimes \sigma_q \\
    & = \sum_{p=1}^d \sum_{q=1}^d ( \alpha^\dagger K^{-1} \alpha)_{pq} \sigma_p \otimes \sigma_q.
\end{align}
Using the definition of spectral norm, we have
\begin{align}
\Tr(A O^U \otimes O^U) & = \sum_{p=1}^d \sum_{q=1}^d ( \alpha^\dagger K^{-1} \alpha)_{pq} \Tr(\sigma_p O^U) \Tr(\sigma_q O^U) \\
& \leq \norm{\alpha^\dagger K^{-1} \alpha}_\infty \sum_{p=1}^d \Tr(O^U \sigma_p)^2.
\end{align}
Recall from the definition below Equation~\eqref{eq:rhosigma}, we have
\begin{equation}
\alpha = U\Sigma V^\dagger, K^{\mathrm{Q}} = \alpha \alpha^\dagger = U \Sigma^2 U^\dagger.    
\end{equation}
Using the fact that orthogonal transformation do not change the spectral norm, $\norm{\alpha^\dagger K^{-1} \alpha}_\infty = \norm{\Sigma U^\dagger K^{-1} U \Sigma}_\infty = \norm{U \Sigma U^\dagger K^{-1} U \Sigma U^\dagger}_\infty = \norm{\sqrt{K^{\mathrm{Q}}} K^{-1} \sqrt{K^{\mathrm{Q}}}}_\infty$.
Hence
\begin{equation}
    \Tr(A O^U \otimes O^U) \leq \norm{\sqrt{K^{\mathrm{Q}}} K^{-1} \sqrt{K^{\mathrm{Q}}}}_\infty \sum_{p=1}^d \Tr(O^U \sigma_p)^2.
\end{equation}
Together with Equation~\eqref{eq:NNprederr}, we have the following prediction error bound 
\begin{equation}
    \mathbb{E}_{x} |h(x) - \Tr(O^U \rho(x))| \leq \mathcal{O}\left(g \sqrt{\frac{\sum_{p=1}^d \Tr(O^U \sigma_p)^2}{N}} + \sqrt{\frac{\log(1 / \delta)}{N}}\right),
\end{equation}
where $g = \sqrt{\norm{\sqrt{K^{\mathrm{Q}}} K^{-1} \sqrt{K^{\mathrm{Q}}}}_\infty}$.
The scalar $g$ measures the closeness of the geometry between the training data points defined by classical neural network and quantum state space.
Note that without the geometric scalar $g$, this prediction error bound is the same as Equation~\eqref{eq:prederrQK} for the quantum kernel method.
Hence, if $g$ is small, classical neural network could predict as well (or potentially better) as the quantum kernel method.
The same analysis in Section~\ref{sec:QKmethodanalysis} allows us to arrive at the following result
\begin{equation}
    \mathbb{E}_{x} |h(x) - \Tr(O^U \rho(x))| \leq \mathcal{O}\left(g \sqrt{\frac{\min(d, \norm{O^U}_F^2)}{N}} + \sqrt{\frac{\log(1 / \delta)}{N}}\right).
\end{equation}
The same analysis holds for other machine learning algorithms, such as Gaussian kernel regression.

\section{Detailed discussion on the relevant quantities s, d, and g}
\label{app:sdgdetail}

There are some important aspects on the three relevant quantities $s, d, g$ that were not fully discussed in the main text, including the limit when we have infinite amount of data and the effect of regularization.  While in  practice one always has a finite amount of data, constructing these formal limits both clarifies the construction and provides another perspective through which to understand the finite data constructions. This section will provide a detailed discussion of these aspects.
 
\subsection{Model complexity s}
\label{sec:modelcomp}

While we have used $s_K(N) = \sum_{i=1}^N \sum_{j=1}^N (K^{-1})_{ij} \Tr(O^U \rho(x_i)) \Tr(O^U \rho(x_j))$ in the main text, this is a simplified quantity when we do not apply regularization. The model complexity $s_K(N)$ under regularization is given by
\begin{align}
s_K(N) = \norm{w}^2 = \Tr(A_{\mathrm{gen}} O^U \otimes O^U) &= \sum_{i=1}^N \sum_{j=1}^N ((K+\lambda I)^{-1} K (K+\lambda I)^{-1})_{ij} \Tr(O^U \rho(x_i)) \Tr(O^U \rho(x_j))\\
&= \sum_{i=1}^N \sum_{j=1}^N (\sqrt{K} (K+\lambda I)^{-2} \sqrt{K})_{ij} \Tr(O^U \rho(x_i)) \Tr(O^U \rho(x_j)).
\end{align}
Training machine learning model with regularization is often desired when we have a finite number $N$ of training data. $\norm{w}^2$ has been used extensively in regularizing machine learning models, see e.g., \cite{krogh1992simple, cortes1995support, suykens1999least}.
This is because we can often significantly reduce generalization error $\sqrt{ \Tr(A_{\mathrm{gen}} O^U \otimes O^U) /N}$ by slightly increasing the training error $\sqrt{ \Tr(A_{\mathrm{tra}} O^U \otimes O^U) / N}$.
In practice, we should choose the regularization parameter $\lambda$ to be a small number such that the training error plus the generalization error is minimized.

The model complexity $s_K(N)$ we have been calculating can be seen as an approximation to the true model complexity when we have a finite number $N$ of training data.
If we have exact knowledge about the input distribution given as a probability measure $\mu_x$, we can also write down the precise model complexity in the reproducing kernel Hilbert space $\phi(x)$ where $k(x, y) = \phi(x)^\dagger \phi(y)$.
Starting from
\begin{equation}
    \min_{w} \lambda w^\dagger w + \int \norm{w^\dagger \phi(x) - \Tr(O^U \rho(x))}^2 d\mu_x,
\end{equation}
we can obtain
\begin{equation}
    w = \left(\lambda I + \int \phi(x)\phi(x)^\dagger d\mu_x\right)^{-1} \int \Tr(O^U \rho(x)) \phi(x) d\mu_x.
\end{equation}
Hence the true model complexity is
\begin{align}
    \norm{w}^2 &= \int \int d\mu_{x_1} d\mu_{x_2} \Tr(O^U \rho(x_1)) \Tr(O^U \rho(x_2)) \phi(x_1)^\dagger \left(\lambda I + \int \phi(\xi)\phi(\xi)^\dagger d\mu_\xi\right)^{-2} \phi(x_2)\\
    &= \Tr(A_{\mathrm{gen}} O^U \otimes O^U),
\end{align}
where the operator
\begin{equation}
    A_{\mathrm{gen}} = \int \int d\mu_{x_1} d\mu_{x_2} \phi(x_1)^\dagger \left(\lambda I + \int \phi(\xi)\phi(\xi)^\dagger d\mu_\xi\right)^{-2} \phi(x_2) \,\, \rho(x_1) \otimes \rho(x_2).
\end{equation}
If we replace the integration over the probability measure with $N$ random samples and apply the fact that $k(x, y) = \phi(x)^\dagger \phi(y)$, then we can obtain the original expression given in Equation~\eqref{eq:Agen}.

\subsection{Dimension d} \label{app:dim-detail}

The dimension we considered in the main text is the effective dimension of the training set quantum state space.
This can be seen as the rank of the quantum kernel matrix $K^Q_{ij} = \Tr(\rho(x_i) \rho(x_j))$.
However, it will often be the case that most of the states lie in some low-dimensional subspace, but have negligible contributions in a much higher dimensional subspace.
In this case, the dimension of the low-dimensional subspace is the better characterization.
More generally, we can perform a singular value decomposition of $K^{\mathrm{Q}}$
\begin{equation} \label{eq:Kqspec}
    K^{\mathrm{Q}} = \sum_{i=1}^N t_i u_i u_i^\dagger,
\end{equation}
with $t_1 \geq t_2 \geq \ldots \geq t_N$.
We define $\sigma_i = \sum_{j=1}^N u_{ij} \rho(x_j) / \norm{\sum_{j=1}^N u_{ij} \rho(x_j)}_F$, where $\norm{\cdot}_F$ is the Frobenius norm. $\sigma_i$ is the $i$-th principal component of the quantum state space.
Recall the normalization condition $\Tr(K^{\mathrm{Q}}) = N$, so $\sum_{i=1}^N t_i = N$.
If the training set quantum state space is one-dimensional ($d = 1$), then
\begin{equation}
    t_1 = N, t_i = 0, \forall i > 1.
\end{equation}
If all the quantum states in the training set are orthogonal ($d = N$), then
\begin{equation}
    t_i = 1, \forall i = 1, \ldots, N.
\end{equation}
By the Eckart-Young-Mirsky theorem, for any $k \geq 1$, the first $k$ principal components $\sigma_1, \ldots, \sigma_k$ form the best $k$-dimensional subspace for approximating $\mathrm{span}\{\rho(x_1), \ldots, \rho(x_N)\}$.
The approximation error is given by
\begin{equation}
    \sum_{i=1}^N \norm{\rho(x_i) - \sum_{j=1}^k \sqrt{t_j} u_{ji} \sigma_j}_F^2 = \sum_{l = k+1}^N t_l.
\end{equation}
As we can see, when the spectrum is flatter, the dimension is larger.
The error decreases at most as $\sum_{l=k}^N t_l \leq N-k$, where the equality holds when all states are orthogonal.
In the numerical experiment, we choose the following measure as the approximate dimension
\begin{equation}
    1 \leq \sum_{k=1}^N \left(\frac{1}{N-k} \sum_{l=k}^N t_l\right) \leq N
\end{equation}
due to the independence to any hyperparameter.
Alternatively, we can also define approximate dimension by choosing the smallest $k$ such that $\sum_{l = k+1}^N t_l / N < \epsilon$ for some $\epsilon > 0$.
Both give similar trend, but the actual value of the dimension would be different.

From the discussion, we can see that in the above definitions, the dimension will always be upper bounded by the number $N$ of training data.
Similar to the case of model complexity, we can also define the dimension $d$ when we have the exact knowledge about the input distribution given by probability measure $\mu_x$.
For a quantum state space representing $n$ qubits, we simply consider the spectrum $t_1 \geq t_2 \geq \ldots \geq t_{2^n}$ of the following operator
\begin{equation}
    \int \mathrm{vec}({\rho}(x)) \mathrm{vec}({\rho}(x))^T d\mu_x.
\end{equation}
When we replace the integration by a finite number of training samples, the spectrum would be equivalent to the spectrum given in Equation~\eqref{eq:Kqspec} except for the additional zeros.

\begin{remark}
The same definition of dimension can be used for any kernels, such as projected quantum kernels or neural tangent kernels (under the normalization $\Tr(K) = N$).
\end{remark}

\subsection{Geometric difference g} \label{app:geo-detail}

The geometric difference is defined between two kernel functions $K^1, K^2$ and the corresponding reproducing kernel Hilbert space $\phi_1(x), \phi_2(x)$.
If we have a function represented by the first kernel $w^\dagger \phi_1(x)$, what would be the model complexity for the second kernel? We consider the ideal case where we know the input distribution $\mu_x$ exactly. The optimization for training the first kernel method with regularization $\lambda > 0$ is
\begin{equation} \label{eq:vopt}
    \min_{v} \lambda v^\dagger v + \int \norm{v^\dagger \phi_2(x) - w^\dagger \phi_1(x)}^2 d\mu_x.
\end{equation}
The solution is given by
\begin{equation}
    v = \left(\lambda I + \int \phi_2(x)\phi_2(x)^\dagger d\mu_x\right)^{-1} \int w^\dagger \phi_1(x) \phi_2(x) d\mu_x.
\end{equation}
Hence the model complexity for the optimized $v$ is
\begin{align}
    \norm{v}^2 &= w^\dagger \left(\int \int d\mu_{x_1} d\mu_{x_2} \phi_1(x_1) \phi_2(x_1)^\dagger \left(\lambda I + \int \phi_2(\xi)\phi_2(\xi)^\dagger d\mu_\xi\right)^{-2} \phi_2(x_2) \phi_1(x_2)^\dagger\right) w\\
    &\leq g_{\mathrm{gen}}^2 \norm{w}^2,
\end{align}
where the geometric difference is
\begin{equation}
    g_{\mathrm{gen}} = \sqrt{\norm{\int \int d\mu_{x_1} d\mu_{x_2} \phi_1(x_1) \phi_2(x_1)^\dagger \left(\lambda I + \int \phi_2(\xi)\phi_2(\xi)^\dagger d\mu_\xi\right)^{-2} \phi_2(x_2) \phi_1(x_2)^\dagger}_\infty}.
\end{equation}
The subscript in $g_{\mathrm{gen}}$ is added because when $\lambda > 0$, there will also be a contribution from training error.
When we only have a finite number $N$ of training samples, we can use the fact that $k(x, y) = \phi(x)^\dagger \phi(y)$ and the definition that $K_{ij} = k(x_i, x_j)$ to obtain
\begin{equation} \label{eq:glambda}
    g_{\mathrm{gen}} = \sqrt{\norm{\sqrt{K^1} \sqrt{K^2} \left(K^2 + \lambda I\right)^{-2} \sqrt{K^2} \sqrt{K^1}}_\infty}.
\end{equation}
This formula differs from the main text due to the regularization parameter $\lambda$. If $\lambda = 0$, then the above formula for $g_{\mathrm{gen}}$ reduces to the formula $g_{\mathrm{gen}} = \sqrt{\norm{\sqrt{K^1} (K^2)^{-1} \sqrt{K^1}}_\infty}$.

When $\lambda$ is non-zero, the geometric difference can become much smaller. This is the same as the discussion on model complexity $s$ in Section~\ref{sec:modelcomp}.
However, a nonzero $\lambda$ induces a small amount of training error.
For a finite number $N$ of samples, the training error can always be upper bounded:
\begin{equation}
    \frac{1}{N}\sum_{i=1}^N \norm{v^\dagger \phi_2(x_i) - w^\dagger \phi_1(x_i)}^2 \leq \lambda^2 \norm{\sqrt{K^1} (K^2 + \lambda I)^{-2} \sqrt{K^1}}_\infty \norm{w}^2 = g_{\mathrm{tra}}^2 \norm{w}^2,
\end{equation}
where $g_{\mathrm{tra}} = \lambda \sqrt{\norm{\sqrt{K^1} (K^2 + \lambda I)^{-2} \sqrt{K^1}}_\infty}$.
This upper bound can be obtained by plugging the solution for $v$ in Equation~\eqref{eq:vopt} under finite samples into the training error $\frac{1}{N}\sum_{i=1}^N \norm{v^\dagger \phi_2(x_i) - w^\dagger \phi_1(x_i)}^2$ and utilizing the fact that $w^\dagger A w \leq \norm{A}_\infty \norm{w}^2$.
In the numerical experiment, we report $g_{\mathrm{gen}}$ given in Equation~\eqref{eq:glambda} with the largest $\lambda$ such that the training error $g_{\mathrm{tra}} \leq 0.045$.

\section{Constructing dataset to separate quantum and classical model}
\label{app:engdata}
In the main text, our central quantity of interest is the geometric difference $g$, which provides a quantification for a given data set, how large the prediction gap can be for possibly function or labels associated with that data.  Here we detail how one can efficiently construct a function that saturates this bound for a given data set.  This is the approach that is used in the main text to engineer the data set with maximal performance.

Given a (projected) quantum kernel $k^{\mathrm{Q}}(x_i, x_j) = \phi^{\mathrm{Q}}(x_i)^\dagger \phi^{\mathrm{Q}}(x_j)$ and a classical kernel $k^{\mathrm{C}}(x_i, x_j) = \phi^{\mathrm{C}}(x_i)^\dagger \phi^{\mathrm{C}}(x_j)$, our goal is to construct a dataset that would best separate the two models.
Consider a dataset $(x_i, y_i), \forall i =1, \ldots, N$.
We use the model complexity $s = \sum_{i=1}^N \sum_{j=1}^N (K^{-1})_{ij} y_i y_j$ to quantify the generalization error of the model.
The model complexity has been introduced in the main text, where a detailed proof relating $s$ to prediction error is given in Appendix~\ref{sec:proofmain}.
To separate between quantum and classical model, we consider $s_{\mathrm{Q}} = 1$ and $s_{\mathrm{C}}$ is as large as possible for a particular choice of targets $y_1, \ldots, y_N$.
To achieve this, we solve the optimization
\begin{equation}
    \min_{y \in \mathbb{R}^N} \frac{\sum_{i=1}^N \sum_{j=1}^N ((K^{\mathrm{C}})^{-1})_{ij} y_i y_j}{\sum_{i=1}^N \sum_{j=1}^N ((K^{\mathrm{Q}})^{-1})_{ij} y_i y_j}
\end{equation}
which has an exact solution given by a generalized eigenvalue problem.  The solution is given by $y = \sqrt{K^{\mathrm{Q}}} v$, where $v$ is the eigenvector of $\sqrt{K^{\mathrm{Q}}} (K^{\mathrm{C}})^{-1} \sqrt{K^{\mathrm{Q}}}$ corresponding to the eigenvalue $g^2 = \norm{\sqrt{K^{\mathrm{Q}}} (K^{\mathrm{C}})^{-1} \sqrt{K^{\mathrm{Q}}}}_\infty$.
This guarantees that $s_{\mathrm{C}} = g^2 s_{\mathrm{Q}} = g^2$, and note that by definition of $g$, $s_{\mathrm{C}} \leq g^2 s_{\mathrm{Q}}$. Hence this dataset fully utilized the geometric difference between the quantum and classical space.

We should also include regularization parameter $\lambda$ when constructing the dataset.
Detailed discussion on model complexity $s$ and geometric difference $g$ with regularization is given in Appendix~\ref{app:sdgdetail}.
Recall that for $\lambda > 0$,
\begin{equation}
    s^\lambda_C = y^\dagger (\sqrt{K^{\mathrm{C}}} \left(K^{\mathrm{C}} + \lambda I\right)^{-2} \sqrt{K^{\mathrm{C}}})_{ij} y,
\end{equation}
which is the model complexity that we want to maximize.
Similar to the unregularized case, we consider the (unregularized) model complexity $s_{\mathrm{Q}} = y^\dagger (K^{\mathrm{Q}})^{-1} y$ to be one.
Solving the generalized eigenvector problem yields the target $y = \sqrt{K^{\mathrm{Q}}} v$, where $v$ is the eigenvector of
\begin{equation}
    \sqrt{K^{\mathrm{Q}}} \sqrt{K^{\mathrm{C}}} \left(K^{\mathrm{C}} + \lambda I\right)^{-2} \sqrt{K^{\mathrm{C}}} \sqrt{K^{\mathrm{Q}}}
\end{equation}
with the corresponding eigenvalue
\begin{equation}
    g_{\mathrm{gen}}^2 = \norm{\sqrt{K^{\mathrm{Q}}} \sqrt{K^{\mathrm{C}}} \left(K^{\mathrm{C}} + \lambda I\right)^{-2} \sqrt{K^{\mathrm{C}}} \sqrt{K^{\mathrm{Q}}}}_\infty.
\end{equation}
The larger $\lambda$ is, the smaller $g_{\mathrm{gen}}^2$ would be.
In practice, one should choose a $\lambda$ such that the training error bound $g_{\mathrm{tra}}^2 s_{\mathrm{Q}} = \lambda^2 \norm{\sqrt{K^{\mathrm{Q}}} (K^{\mathrm{C}} + \lambda I)^{-2} \sqrt{K^{\mathrm{Q}}}}_\infty$ for the classical ML model is small enough. In the numerical experiment, we choose a $\lambda$ such that the training error bound $g_{\mathrm{tra}}^2 s_{\mathrm{Q}} \leq 0.002$ and $g_{\mathrm{gen}}$ is as large as possible.
Finally, we can turn this dataset, which maps input $x$ to a real value $y$, into a classification task by replacing $y$ with $+1$ if $y > \mathrm{median}(y_1, \ldots, y_N)$ and $-1$ if $y \leq \mathrm{median}(y_1, \ldots, y_N)$.

The constructed dataset will yield the largest separation between quantum and classical models from a learning-theoretic sense, as the model complexity fully saturates the geometric difference.
If there is no quantum advantage in this dataset, there will likely be none.
We believe this construction procedure will eventually lead to the first quantum advantage in machine learning problems (classification problems to be more specific).

\section{Lower bound on learning quantum models}
\label{app:lowerbound}

In this section, we will prove a fundamental lower bound for learning quantum models stated in Theorem~\ref{thm:lowerbound}. This result says that in the worst case, the number $N$ of training data has to be at least $\Omega(\Tr(O^2) / \epsilon^2)$ when the input quantum state can be distributed across a sufficiently large Hilbert space. Quantum kernel method matches this lower bound. When the data spans over the entire Hilbert space, the dimension $d$ will be large and the prediction error of the quantum kernel method given in Equation~\eqref{eq:predqkm} becomes
\begin{equation}
    \mathbb{E}_{x} |h^{\mathrm{Q}}(x) - \Tr(O^U \rho(x))| \leq \mathcal{O}\left( \sqrt{\frac{ \Tr(O^2)}{N}}\right).
\end{equation}
Hence we can achieve $\epsilon$ error using $N \leq \mathcal{O}(\Tr(O^2) / \epsilon^2)$ matching the fundamental lower bound.

\begin{theorem} \label{thm:lowerbound}
Consider any learning algorithm $\mathcal{A}$. Suppose for any unknown unitary evolution $U$, any unknown observable $O$ with bounded Frobenius norm $\Tr(O^2) \leq B$, and any distribution $D$ over the input quantum states, the learning algorithm $\mathcal{A}$ could learn a function $h$ such that
\begin{equation}
    \E_{\rho \sim D} |h(\rho) - \Tr(O U \rho U^\dagger)| \leq \epsilon,
\end{equation}
from $N$ training data $(\rho_i, \Tr(O U \rho_i U^\dagger)), \forall i = 1, \ldots, N$ with high probability.
Then we must have
\begin{equation}
    N \geq \Omega(B / \epsilon^2).
\end{equation}
\end{theorem}
\begin{proof}
We select a Hilbert space with dimension $d = B / 4 \epsilon^2$ (this could be a subspace of an exponentially large Hilbert space). We define the distribution $D$ to be the uniform distribution over the basis states $\ket{x}\!\bra{x}$ of the $d$-dimensional Hilbert space.
Then we consider the unknown unitary $U$ to always be the identity, while the possible observables are
\begin{equation}
    O_v = 2 \epsilon \sum_{x = 1}^{d} v_x \ket{x}\!\bra{x},
\end{equation}
with $v_x \in \{\pm 1\}, \forall x = 1, \ldots, d$.
There are hence $2^d$ different choices of observables $O_v$.

We now set up a simple communication protocol to prove the lower bound on the number of data needed. This is a simplified version of the proofs found in Refs.\cite{haah2017sample, Huang2020}.
Alice samples an observable $O_v$ uniformly at random from the $2^d$ possible choices.
We can treat $v$ as a bit-string of $d$ entries.
Then she samples $N$ quantum states $\ket{x_i}\!\bra{x_i}, \forall i = 1, \ldots, N$.
Alice then gives Bob the following training data $\mathcal{T} = \{(\ket{x_i}\!\bra{x_i}, \bra{x_i}O_v \ket{x_i}) = v_{x_i}, \forall i = 1, \ldots, N\}$.
Notice that the mutual information $I(v, \mathcal{T})$ between $v$ and the training data $\mathcal{T}$ satisfies
\begin{equation}
I(s, \mathcal{T}) \leq N,\label{eq:Iupper}    
\end{equation}
because the training data contains at most $N$ values of $s$.

With high probability, the following is true by the requirement of the learning algorithm $\mathcal{A}$.
Using the training data $\mathcal{T}$, Bob can apply the learning algorithm $\mathcal{A}$ to obtain a function $f$ such that
\begin{equation}
    \E_{\rho \sim D} |h(\rho) - \Tr(O U \rho U^\dagger)| \leq \epsilon.
\end{equation}
Using Markov's inequality, we have
\begin{equation}
    \mathrm{Pr}[ |h(\rho) - \Tr(O U \rho U^\dagger)| < 2 \epsilon ] > \frac{1}{2}. 
    \label{eq:probf}
\end{equation}
For all $x = 1, \ldots, d$, if $|h(\ket{x}\!\bra{x}) - \Tr(O U \ket{x}\!\bra{x} U^\dagger)| < 2 \epsilon$, we have $|h(\ket{x}\!\bra{x}) / 2 \epsilon - v_{x}| < 1$.
This means that if $h(\ket{x}\!\bra{x}) > 0$, then $v_x = 1$ and if $h(\ket{x}\!\bra{x}) < 0$, then $v_x = -1$.
Hence Bob can construct a bit-string $\tilde{v}$ given as $\tilde{v}_x = \mathrm{sign}(h(\ket{x}\!\bra{x})), \forall x = 1, \ldots, d$. Using Equation~\eqref{eq:probf}, we know that at least $d / 2$ bits in $\tilde{v}$ will be equal to $v$.

Because with high probability, $\tilde{v}$ and $v$ has at least $d / 2$ bits in common. Fano's inequality tells us that $I(v, \tilde{v}) \geq \Omega(d)$. Because the bit-string $\tilde{v}$ is constructed solely from the training data $\mathcal{T}$. Data processing inequality tells us that $I(v, \tilde{v}) \leq I (v, \mathcal{T})$.
Together with Equation~\eqref{eq:Iupper}, we have
\begin{equation}
    N \geq I (v, \mathcal{T}) \geq I(v, \tilde{v}) \geq \Omega(d).
\end{equation}
Recall that $d = B / 4\epsilon^2$, we have hence obtained the desired result $N \geq \Omega(B / \epsilon^2)$.
\end{proof}

\section{Limitations of quantum kernel methods}
\label{app:limitqk}

Even though the quantum kernel method saturates the fundamental lower bound $\Omega(\Tr(O^2) / \epsilon^2)$ and can be made formally equivalent to infinite depth quantum neural networks it has a number of limitations that hinder its practical applicability. In this section we construct a simple example where the overhead for using the quantum kernel method is exponential in comparison to trivial classical methods.  

Specifically, it has the limitation of closely following this lower bound for any unitary $U$ and observable $O$. This is not true for other machine learning methods, such as classical neural networks or projected quantum kernel methods.
It is possible for classical machine learning methods to learn quantum models with exponentially large $\Tr(O^2)$, which is not learnable by the quantum kernel method.
This can already be seen in the numerical experiments given in the main text.
In this section, we provide a simple example that allows theoretical analysis to illustrate this limitation.

We consider a simple learning task where the input vector $x \in \{0, \pi\}^n$.
The encoding of the input vector $x$ to the quantum state space is given as
\begin{equation}
    \ket{x} = \prod_{k=1}^n \exp(\mathrm{i} X_k x_k) \ket{0^n}.
\end{equation}
The quantum state $\ket{x}$ is a computational basis state. We define $\rho(x) = \ket{x}\!\bra{x}$.
The quantum model applies a unitary $U = I$, and measures the observable $O = I \otimes \ldots \otimes I \otimes Z$.
Hence $f(x) = \Tr(O \rho(x)) = (2 x_n - \pi)$.
Notice that for this very simple quantum model, the function $f(x)$ is an extremely simple linear model.
Hence a linear regression or a single-layer neural network can learn the function $f(x)$ from training data of size $n$ with high probability.

Despite being a very simple quantum model, the Frobenius norm of the observable $\Tr(O^2)$ is exponentially large, i.e., $\Tr(O^2) = 2^n$.
We now show that a quantum kernel method will need a training data of size $N \geq \Omega(2^n)$ to learn this simple function $f(x)$.
Suppose we have obtained a training set given as $\{(x_i, \Tr(O \rho(x_i))\}_{i=1}^N$ where each $x_i$ is selected uniformly at random from $\{0, \pi\}^n$. Recall from the analysis in Section~\ref{sec:MLmodels}, the function learned by the quantum kernel method will be
\begin{equation}
    h^{\mathrm{Q}}(x) = \min\left(1, \max\left(-1, \sum_{i=1}^N \sum_{j=1}^N \Tr(\rho(x_i) \rho(x)) ((K^{\mathrm{Q}}+\lambda I)^{-1})_{ij} \Tr(O \rho(x_j))\right)\right),
\end{equation}
where $K^Q_{ij} = k^{\mathrm{Q}}(x_i, x_j) = \Tr(\rho(x_i) \rho(x_j))$.
The main problem of the quantum kernel method comes from the precise definition of the kernel function $k(x_i, x) = \Tr(\rho(x_i) \rho(x))$.
For at least $2^n - N$ choices of $x$, we have $\Tr(\rho(x_i) \rho(x)) = 0, \forall i = 1, \ldots, N$.
This means that for at least $2^n - N$ choices of $x$, $h^{\mathrm{Q}}(x) = 0$.
However, by construction, $f(x) \in \{1, -1\}$.
Hence the prediction error can be lower bounded by
\begin{equation}
    \frac{1}{2^n} \sum_{x \in \{0, \pi\}^n} |h^{\mathrm{Q}}(x) - f(x)| \geq 1 - \frac{N}{2^n}.
\end{equation}
Therefore, if $N < (1 - \epsilon) 2^n$, then the prediction error will be greater than $\epsilon$. Hence we need a training set of size $N \geq (1 - \epsilon) 2^n$ to achieve a prediction error $\leq \epsilon$.

In general, when we place the classical vectors $x$ into an exponentially large quantum state space, the quantum kernel function $\Tr(\rho(x_i) \rho(x_j))$ will be exponentially close to zero for $x_i \neq x_j$. In this case $K^{\mathrm{Q}}$ will be close to the identity matrix, but $\Tr(\rho(x_i) \rho(x))$ will be exponentially small. For a training set of size $N \ll 2^n$, $h^{\mathrm{Q}}(x)$ will be exponentially close to zero similar to the above example.
Despite $h^{\mathrm{Q}}(x)$ being exponentially close to zero, if we can distinguish $> 0$ and $< 0$, then $h^{\mathrm{Q}}$ could still be useful in classification tasks.
However, due to the inherent quantum measurement error in evaluating the kernel function $\Tr(\rho(x_i) \rho(x_j))$ on a quantum computer, we will need an exponential number of measurements to resolve such an exponentially small difference.

\section{Projected quantum kernel methods}\label{app:pqk}

In the main text, we argue that projection back from the quantum space to a classical one in the projected quantum kernel can greatly improve the performance of such methods.  There we focused on the simple case of a squared exponential based on reduced 1-particle observables, however this idea is far more general.  In this section we explore some of these generalizations including a novel scheme for calculating functions of all powers of RDMs efficiently.

From discussions on the quantum kernel method, we have seen that using the native quantum state space to define the kernel function, e.g., $k(x_i, x_j) = \Tr(\rho(x_i) \rho(x_j))$ can fail to learn even a simple function when the full exponential quantum state space is being used.
We have to utilize the entire exponential quantum state space otherwise the quantum machine learning model could be simulated efficiently classically and a large advantage could not be found.
In this section, we will detail a set of solutions that project the quantum states back to approximate classical representations and define the kernel function using the classical representation. We refer to these modified quantum kernels as projected quantum kernels.
The projected quantum kernels are defined in a classical vector space to circumvent the hardness of learning due to the exponential dimension in quantum Hilbert space.
However, projected quantum kernels still use the exponentially large quantum Hilbert space for evaluation and can be hard to simulate classically.

Some simple choices based on reduced density matrices (RDMs) of the quantum state are given below.
\begin{enumerate}
    \item A linear kernel function using 1-RDMs
    \begin{align}
        Q_l^1(x_i, x_j) = \sum_k \text{Tr} \left[ \text{Tr}_{m \neq k} [\rho(x_i)] \text{Tr}_{n \neq k}[ \rho(x_j)] \right],
    \end{align}
    where $\Tr_m \neq k(\rho)$ is the partial trace of the quantum state $\rho$ over all qubits except for the $k$-th qubit. It could learn any observable that can be written as a sum of one-body terms.
    \item A Gaussian kernel function using 1-RDMs
    \begin{align} \label{eq:gau1rdm}
        Q_g^1(x_i, x_j) = \exp \left( -\gamma  \sum_k \left(\text{Tr}_{m \neq k}[\rho(x_i)] - \text{Tr}_{n \neq k}[\rho(x_j)] \right)^2 \right),
    \end{align}
    where $\gamma > 0$ is a hyper-parameter. It could learn any nonlinear function of the $1$-RDMs. 
    \item A linear kernel using $k-$RDMs
    \begin{align}
        Q_l^k(x_i, x_j) = \sum_{K \in S_k(n)}\text{Tr} \left[ \text{Tr}_{n \notin K}[\rho(x_i)] \text{Tr}_{m \notin K}[\rho(x_j)] \right]
    \end{align}
    where $S_k(n)$ is the set of subsets of $k$ qubits from $n$, $\text{Tr}_{n \notin K}$ is a partial trace over all qubits not in subset $K$. It could learn any observable that can be written as a sum of $k$-body terms.
\end{enumerate}
The above choices have a limited function class that they can learn, e.g., $Q_l^1$ can only learn observables that are sum of single-qubit observables.
It is desirable to define a kernel that can learn any quantum models (e.g., arbitrarily deep quantum neural networks) with sufficient amount of data similar to the original quantum kernel $k^{\mathrm{Q}}(x_i, x_j) = \Tr(\rho(x_i) \rho(x_j))$ as discussed in Appendix~\ref{sec:qk_qnn}.

We now define a projected quantum kernel that contains all orders of RDMs. Since all quantum models $f(x) = \Tr(O U \rho(x) U^\dagger)$ are linear functions of the full quantum state, this kernel can learn any quantum models with sufficient data.
A $k$-RDM of a quantum state $\rho(x)$ for qubit indices $(p_1, p_2, \ldots, p_k)$ can be reconstructed by local randomized measurements using the formalism of classical shadows \cite{Huang2020}:
\begin{equation}
    \rho^{(p_1, p_2, \ldots, p_k)}(x) = \mathbb{E}\left[ \otimes_{r=1}^k (3 \ket{s_{p_r}, b_{p_r}}\!\bra{s_{p_r}, b_{p_r}} - I) \right],
\end{equation}
where $b_{p_r}$ is a random Pauli measurement basis $X, Y, Z$ on the $p_r$-th qubit, and $s_{p_r}$ is the measurement outcome $\pm 1$ on the $p_r$-th qubit of the quantum state $\rho(x)$ under Pauli basis $b_{p_r}$.
The expectation is taken with respect to the randomized measurement on $\rho(x)$.
The inner product of two $k$-RDMs is equal to
\begin{equation}
\label{eq:innshadow}
    \mathrm{Tr}\left[\rho^{(p_1, p_2, \ldots, p_k)}(x_i) \rho^{(p_1, p_2, \ldots, p_k)}(x_j)]\right] = \mathbb{E} \left[ \Pi_{r=1}^k (9 \delta_{s^i_{p_r} s^j_{p_r}} \delta_{b^i_{p_r} b^j_{p_r}} - 4) \right],
\end{equation}
where we used the fact that the randomized measurement outcomes for $\rho(x_i)$ and $\rho(x_j)$ are independent.
We extend this equation to the case where some indices $p_r, p_s$ coincide. This would only introduce additional features in the feature map $\phi(x)$ that defines the kernel $k(x_i, x_j) = \phi(x_i)^\dagger \phi(x_j)$.
The sum of all possible $k$-RDMs can be written as
\begin{equation}
    Q^k(\rho(x_i), \rho(x_j)) = \sum_{p_1=1}^n \ldots \sum_{p_k=1}^n \mathrm{Tr}\left[\rho^{(p_1, p_2, \ldots, p_k)}(x_i) \rho^{(p_1, p_2, \ldots, p_k)}(x_j)]\right] = \mathbb{E} \left[\left( \sum_{p=1}^n (9 \delta_{s^i_{p} s^j_{p}} \delta_{b^i_{p} b^j_{p}} - 4) \right)^k \right],
\end{equation}
where we used Equation~\eqref{eq:innshadow} and linearity of expectation.
A kernel function that contains all orders of RDMs can be evaluated as
\begin{equation}
    Q^\infty_\gamma(\rho(x_i), \rho(x_j)) = \sum_{k=0}^\infty \frac{\gamma^k}{k! n^k} Q^k(\rho(x_i), \rho(x_j)) = \mathbb{E} \exp\left( \frac{\gamma}{n} \sum_{p=1}^n (9 \delta_{s^i_{p} s^j_{p}} \delta_{b^i_{p} b^j_{p}} - 4) \right),
\end{equation}
where $\gamma$ is a hyper-parameter.
The kernel function $Q^\infty_\gamma(\rho(x_i), \rho(x_j))$ can be computed by performing local randomized measurement on the quantum states $\rho(x_i)$ and $\rho(x_j)$ independently.
First, we collect a set of randomized measurement data for $\rho(x_i), \rho(x_j)$ independently:
\begin{align}
    \rho(x_i) &\rightarrow \{((s^{i, r}_1, b^{i, r}_1), \ldots , (s^{i, r}_n, b^{i, r}_n)), \forall r = 1, \ldots, N_s\},\\
    \rho(x_j) &\rightarrow \{((s^{j, r}_1, b^{j, r}_1), \ldots , (s^{j, r}_n, b^{j, r}_n)), \forall r = 1, \ldots, N_s\},
\end{align}
where $N_s$ is the number of repetition for each quantum state. For each repetition, we will randomly sample a Pauli basis for each qubit and measure that qubit to obtain an outcome $\pm1$. For the $r$-th repetition, the Pauli basis in the $k$-th qubit is given as $b^{i, r}_k$ and the measurement outcome $\pm 1$ is given as $s^{i, r}_k$.
Then we compute
\begin{equation}
    \frac{1}{N_s(N_s - 1)} \sum_{r_1 = 1}^{N_s} \sum_{\substack{r_2 = 1\\ r_2 \neq r_1}}^{N_s} \exp\left( \frac{\gamma}{n} \sum_{p=1}^n (9 \delta_{s^{i, r_1}_{p} s^{j, r_2}_{p}} \delta_{b^{i, r_1}_{p} b^{j, r_2}_{p}} - 4) \right) \approx Q^\infty_\gamma(\rho(x_i), \rho(x_j)).
\end{equation}
We reuse all pairs of data $r_1, r_2$ to reduce variance when estimating $Q^\infty_\gamma(\rho(x_i)$, since the resulting estimator would still be equal to the desired quantity in expectation. This technique is known as U-statistics, which is often used to create minimum-variance unbiased estimators. U-statistics is also applied in \cite{Huang2020} for estimating Renyi entanglement entropy with high accuracy.


\section{Simple and rigorous quantum advantage over classical machine learning models} \label{app:qadv}

In Ref.\cite{liu2020rigorous}, the authors proposed a machine learning problem based on discrete logarithm which is assumed to be hard for any classical machine learning algorithm, complementing existing work studying learnability in the context of discrete logs~\cite{servedio2004equivalences,sweke2020quantum}.  Much of the challenge in their construction \cite{sweke2020quantum} was related to technicalities involved in the original quantum kernel approach.  Here we present a simple quantum machine learning algorithm using the projected quantum kernel method.
The problem is defined as follows, where $p$ is an exponentially large prime number and $g$ is chosen such that computing $\log_g(x)$ in $\mathbb{Z}_p^*$ is classically hard and $\log_g(x)$ is one-to-one.

\begin{definition}[Discrete logarithm-based learning problem]
For all input $x \in \mathbb{Z}_p^*$, where $n = \lceil \log_2(p) \rceil$, the output is
\begin{equation}
    y(x) = \begin{cases} +1, & \log_g(x) \in [s, s+\frac{p-3}{2}], \\ -1, & \log_g(x) \notin [s, s+\frac{p-3}{2}], \end{cases}
\end{equation}
for some $s \in \mathbb{Z}_p^*$. The goal is to predict $y(x)$ for an input $x$ sampled uniformly from $\mathbb{Z}_p^*$.
\end{definition}

Let us consider the most straight-forward feature mapping that maps the classical input $x$ into the quantum state space $\ket{\log_g(x)}$ using Shor's algorithm for computing discrete logarithms \cite{nielsen2001quantum}.

Training the original quantum kernel method using this feature mapping will require training data $\{(x_i, y_i)\}_{i=1}^N$ with $N$ being exponentially large to yield a small prediction error.
This is because for a new $x \in \mathbb{Z}_p^*$, such that $\log_g(x) \neq \log_g(x_i), \forall i = 1, \ldots, N$, quantum kernel method will be equivalent to random guessing.
Hence the quantum kernel method has to see most of the values in the range of $\log_g(x)$ ($\mathbb{Z}_p^*$) to make accurate predictions. This is the same as the example to demonstrate the limitation of quantum kernel methods in Appendix~\ref{app:limitqk}.
Since $\mathbb{Z}_p^*$ is exponentially large, the quantum kernel method has to use an exponentially amount number of data $N$ for this straight-forward feature map.
The central problem is that all the inputs $x$ are maximally far apart from one another, and this impedes the ability for quantum kernel methods to generalize.

On the other hand, we can project the quantum feature map $\ket{\log_g(x)}$ back to a classical space, which is now just a number $\log_g(x) \in \mathbb{Z}_p^*$.
Recall that $\mathbb{Z}_p^*$ contains all number from $0, \ldots, p-1$, thus we consider mapping $x$ to a real number $z = \log_g(x) / p \in [0, 1)$.
Let us define $t = s / p$.
In this projected space, we are learning a simple classification problem where $y(z) = +1$ if $z \in [t, t + \frac{p-3}{2p}]$, and $y(z) = -1$ if $z \notin [t, t + \frac{p-3}{2p}]$.
We are using a periodic boundary where $0$ and $1$ are the same point.
If $t + \frac{p-3}{2p} < 1$, then there exists some $a, b \in [0, 1)$ and $a < b$, such that $y(z) = +1$, if $a \leq z \leq b$, and $y(z) = -1$, otherwise.
In this case we have $y(z) = \mathrm{sign}((b-z)(z-a))$, where $\mathrm{sign}(t) = +1$ if $t \geq 0$, otherwise $\mathrm{sign}(t) = -1$.
If $t + \frac{p-3}{2p} \geq 1$, then there exists some $a, b \in [0, 1)$ and $a < b$, such that $y(z) = -1$, if $a \leq z \leq b$, and $y(z) = +1$, otherwise.
In this case we have $y(z) = \mathrm{sign}((a-z)(z-b))$.
Through this analysis, we can see that we only need to learn a simple quadratic function to perform accurate classification.
Hence one could simply define a projected quantum kernel as
\begin{equation}
    k^{\mathrm{PQ}}(x_i, x_j) = \left((\log_g(x_i) / p) (\log_g(x_j) / p) + 1\right)^2,
\end{equation}
where the division in $(\log_g(x_i) / p)$ is performed as real number in $\mathbb{R}$.
This projected quantum kernel can efficiently learn any quadratic function $az^2 + bz + c$ with $z = \log_g(x_i)/p$, hence solving the above learning problem.

\begin{theorem}[Corollary 3.19 in \cite{mohri2018foundations}] \label{thm:VCdim}
Let $\mathcal{H}$ be a class of functions taking values in $\{+1, -1\}$ with VC-dimension $d$. Then with probability $\geq 1 - \delta$ over the sampling of $z_1, \ldots z_N$ from some distribution $\mathcal{D}$, we have
\begin{equation}
    \mathop{\mathbb{E}}_{z \sim \mathcal{D}} I[h(z) \neq y(z)] \leq \frac{1}{N} \sum_{i=1}^N I[h(z_i) \neq y(z_i)] + \sqrt{\frac{2d \log(\mathrm{e}N / d)}{N}} + \sqrt{\frac{\log(1 / \delta)}{N}}, \label{eq:vcdimgen}
\end{equation}
for all $h \in \mathcal{H}$, where $I[\text{Statement}] = 1$ if $\text{Statement}$ is true, otherwise $I[\text{Statement}] = 0$.
\end{theorem}

A simple and rigorous statement could be made by noticing that the VC-dimension \cite{blumer1989learnability, mohri2018foundations} for the function class $\{ \mathrm{sign}(az^2 + bz + c) | a, b, c \in \mathbb{R}\}$ is $3$. Let us apply Theorem~\ref{thm:VCdim} with
\begin{equation}
 z = \log_g(x) / p \,\,\, \mathrm{ and } \,\,\, \mathcal{H} = \{\mathrm{sign}(az^2 + bz + c) | a, b, c \in \mathbb{R}\}.
\end{equation}
This theorem bounds the prediction error for new inputs $z$ coming from the same distribution as how the training data is sampled.
For a given set of training data $(z_i, y(z_i))_{i=1}^N$, we perform a minimization over $a, b, c \in \mathbb{R}$ such that the training error $\frac{1}{N} \sum_{i=1}^N I[h(z_i) \neq y(z_i)]$ is zero. This can be achieved by applying a standard support vector machine algorithm \cite{chang2011libsvm} using the above kernel $k^{\mathrm{PQ}}$, because $y(z_i) \in \mathcal{H}$, so one can always fit the training data perfectly. Using Eq.~\eqref{eq:vcdimgen} with $\delta = 0.01$, we can provide a prediction error bound for the trained projected quantum kernel method
\begin{equation}
f_*(x) = h_*(\log_g(x) / p) = h_*(z) = \mathrm{sign}(a_*z^2 + b_*z + c_*).
\end{equation}
Because we fit the training data perfectly, we have
\begin{equation}
    \frac{1}{N} \sum_{i=1}^N I[h_*(z_i) \neq y(z_i)] = 0.
\end{equation}
With probability at least $0.99$, a projected quantum kernel method $f_*(x) = h_*(\log_g(x) / p)$ that perfectly fit a data set of size $N = \mathcal{O}(\log(1 / \epsilon) / \epsilon^2)$ has a prediction error
\begin{equation}
    \mathop{\mathbb{P}}_{x \sim \mathbb{Z}_p^*} [f(x) \neq y(x)] \leq \epsilon.
\end{equation}
This concludes the proof showing that the discrete logarithm-based learning problem can be solved with a projected quantum kernel method using a sample complexity independent of the input size $n$.

Despite the limitations of the quantum kernel method, the authors in \cite{liu2020rigorous} have shown that a clever choice of feature mapping $x \rightarrow \rho(x)$ would also allow quantum kernels $\Tr(\rho(x_i) \rho(x_j))$ to predict well in this learning problem. 

\section{Details of numerical studies}
\label{app:detailnum}

Here we give the complete details for the numerical studies presented in the main text. For the input distribution, we focused on the fashion MNIST dataset \cite{xiao2017fashion}. We use principal component analysis (PCA) provided by scikit-learn \cite{sklearn_api} to map each image ($28 \times 28$ grayscale) into classical vectors $x_i \in \mathbb{R}^n$, where $n$ is the number of principal components.
After PCA, we normalize the vectors $x_i$ such that each dimension is centered at zero and the standard deviation is one.
Finally, we sub-sample $800$ data points from the dataset without replacement.

\subsection{Embedding classical data into quantum states}
\label{sec:embeddingdetail}

The three approaches for embedding classical vectors $x_i \in \mathbb{R}^n$ into quantum states $\ket{x_i}$ are given below. 
\begin{itemize}
    \item \textbf{E1}: Separable encoding or qubit rotation circuit. This is a common choice in literature, e.g., see \cite{schuld2019quantum, skolik2020layerwise}.
    \begin{equation}
        \ket{x_i} = \bigotimes_{j=1}^n \mathrm{e}^{-\mathrm{i} X_j x_{ij}} \ket{0^n},
    \end{equation}
    where $x_{ij}$ is the $j$-th entry of the $n$-dim. vector $x_i$, $X_j$ is the Pauli-X operator acting on the $j$-th qubit.
    \item \textbf{E2}: IQP-style encoding circuit. This is an embedding proposed in \cite{havlivcek2019supervised} that suggests a quantum advantage.
    \begin{equation}
        \ket{x_i} = U_Z(x_i) H^{\otimes n} U_Z(x_i) H^{\otimes n} \ket{0^n},
    \end{equation}
    where $H^{\otimes n}$ is the unitary that applies Hadamard gates on all qubits in parallel, and
    \begin{equation}
        U_Z(x_i) = \exp\left( \sum_{j=1}^n x_{ij} Z_j + \sum_{j=1}^n \sum_{j'=1}^n x_{ij} x_{ij'} Z_j Z_{j'} \right),
    \end{equation}
    with $Z_j$ defined as the Pauli-Z operator acting on the $j$-th qubit.
    In the original proposal \cite{havlivcek2019supervised}, $x \in [0, 2 \pi]^n$,
    and they used $U_Z(x_i) = \exp\left( \sum_{j=1}^n x_{ij} Z_j + \sum_{j=1}^n \sum_{j'=1}^n (\pi - x_{ij})(\pi - x_{ij'}) Z_j Z_{j'} \right)$ instead.
    Here, due to the data pre-processing steps, $x$ will be centered around $0$ with a standard deviation of $1$, hence we made the equivalent changes to the definition of $U_Z(x_i)$.
    \item \textbf{E3}: A Hamiltonian evolution ansatz. This ansatz has been explored in the literature \cite{wecker2015progress, cade2019strategies, wiersema2020exploring} for quantum many-body problems. We consider a Trotter formula with $T$ Trotter steps (we choose $T = 20$) for evolving an 1D-Heisenberg model with interactions given by the classical vector $x_i$ for a time $t$ proportional to the system size (we choose $t = n / 3$).
    \begin{equation}
        \ket{x_i} = \left( \prod_{j=1}^n \exp\left(-\mathrm{i} \frac{t}{T} x_{ij} \left(X_j X_{j+1} + Y_j Y_{j+1} + Z_j Z_{j+1}\right)\right) \right)^T \bigotimes_{j=1}^{n+1} \ket{\psi_j},
    \end{equation}
    where $X_j, Y_j, Z_j$ are the Pauli operators for the $j$-th qubit and $\ket{\psi_j}$ is a Haar-random single-qubit quantum state.
    We sample and fix the Haar-random quantum states $\ket{\psi_j}$ for every qubit.
\end{itemize}

\subsection{Definition of original and projected quantum kernels}

We use Tensorflow-Quantum \cite{broughton2020tensorflow} for implementing the original/projected quantum kernel methods. This is done by performing quantum circuit simulation for the above embeddings and computing the kernel function $k(x_i, x_j)$.
For quantum kernel, we store the quantum states $\ket{x_i}$ as explicit amplitude vectors and compute the squared inner product 
\begin{equation}
    k^{\text{Q}}(x_i, x_j) = |\braket{x_i}{x_j}|^2.
\end{equation}
On actual quantum computers, we obtain the quantum kernel by measuring the expectation of the observable $\ket{0^n}\!\bra{0^n}$ on the quantum state $U_{\mathrm{emb}}(x_j)^\dagger U_{\mathrm{emb}}(x_i) \ket{0^n}$.
For projected quantum kernel, we use the kernel function
\begin{align} \label{eq:PQKnumerics}
    k^{\text{PQ}}(x_i, x_j) = \exp \left(-\gamma  \sum_k \sum_{P \in \{X, Y, Z\}} \left(\Tr(P \rho(x_i)_k) - \Tr( P \rho(x_j)_k)\right)^2 \right),
\end{align}
where $P$ is a Pauli matrix and $\gamma > 0$ is a hyper-parameter chosen to maximize prediction accuracy. We compute the kernel matrix $K \in \mathbb{R}^{N \times N}$ with $K_{ij} = k(x_i, x_j)$ using the sub-sampled dataset with $N=800$ for both the original/projected quantum kernel.

\subsection{Dimension and geometric difference}

Following the discussion in Appendix~\ref{app:dim-detail},
the approximate dimension of the original/projected quantum space is computed by
\begin{equation}
    \sum_{k=1}^N \left(\frac{1}{N-k} \sum_{l=k}^N t_l\right),
\end{equation}
where $N = 800$ and $t_1 \geq t_2 \geq \ldots \geq t_N$ are the singular values of the kernel matrix $K \in \mathbb{R}^{N \times N}$.
Based on the discussion in Appendix~\ref{app:geo-detail},
we report the minimum geometric difference $g$ of the original/projected quantum space (we refer to both the original/projected quantum kernel matrix as $K^{\text{P/Q}}$)
\begin{equation}
    g_{\mathrm{gen}} = \sqrt{\norm{\sqrt{K^{\text{P/Q}}} \sqrt{K^{\mathrm{C}}} \left(K^{\mathrm{C}} + \lambda I\right)^{-2} \sqrt{K^{\mathrm{C}}} \sqrt{K^{\text{P/Q}}}}_\infty},
\end{equation}
under a condition for having a small training error
\begin{equation}
    g_{\mathrm{tra}} = \lambda \sqrt{\norm{\sqrt{K^{\text{P/Q}}} (K^{\mathrm{C}} + \lambda I)^{-2} \sqrt{K^{\text{P/Q}}}}_\infty} < 0.045.
\end{equation}
The actual value of $g$ will depend on the list of choices for $\lambda$ and classical kernels $K^{\mathrm{C}}$.
We consider the following list of $\lambda$
\begin{equation}
    \lambda \in \{0.00001, 0.0001, 0.001, 0.01, 0.025, 0.05, 0.1\},
\end{equation}
and classical kernel matrix $K^{\mathrm{C}}$ being the linear kernel $k^l(x_i, x_j) = x_i^\dagger x_j$ or the Gaussian kernel $k^\gamma(x_i, x_j) = \exp(-\gamma \norm{x_i - x_j}^2)$ with hyper-parameter $\gamma$ from the list
\begin{equation}
    \gamma \in \{0.25, 0.5, 1.0, 2.0, 4.0, 8.0, 16.0, 32.0, 64.0\} / (n \Var[x_{ik}]) 
\end{equation}
for estimating the minimum geometric difference.
$\Var[x_{ik}]$ is the variance of all the coordinates $k=1, \ldots, n$ from all the data points $x_1, \ldots, x_N$.
One could add more choices of regularization parameters $\lambda$ or classical kernel functions, such as using polynomial kernels or neural tangent kernels, which are equivalent to training neural networks with large hidden layers (a package, called Neural Tangents \cite{neuraltangents2020}, is available for use).
This will provide a smaller geometric difference with the quantum state space, but all theoretical predictions remain unchanged.

\subsection{Datasets}

We include a variety of classical and quantum data sets.

\begin{enumerate}
    \item \textbf{Dataset (C)}: For the original classical image recognition data set, i.e., Dataset (C) in Figure~\ref{fig:DimensionAccuracy}(b), we choose two classes, dresses (class $3$) and shirts (class $6$), to form a binary classification task. The prediction error (between 0.0 and 1.0) is equal to the portion of data that are incorrectly labeled.
    \item \textbf{Dataset (Q, E1/E2/E3)}: For the quantum data sets in Figure~\ref{fig:DimensionAccuracy}(b), we consider the following quantum neural network
    \begin{equation}
        U_{\mathrm{QNN}} = \left( \prod_{j=1}^n \exp\left(-\mathrm{i} \frac{t}{T} J_{j} \left(X_j X_{j+1} + Y_j Y_{j+1} + Z_j Z_{j+1}\right)\right) \right)^T,
    \end{equation}
    where we choose $T = t = 10$ and $J_{j} \in \mathbb{R}$ are randomly sampled from the Gaussian distribution with mean $0$ and standard deviation $1$.
    We measure $Z_1$ after the quantum neural network, hence the resulting function is
    \begin{equation}
        f(x) = \Tr(Z_1 U_{\mathrm{QNN}} \ket{x}\!\bra{x} U_{\mathrm{QNN}}^\dagger).
    \end{equation}
    The mapping from $x$ to $\ket{x}$ depends on the feature embedding (E1, E2, or E3) discussed in Section~\ref{sec:embeddingdetail}. A different embedding $\ket{x}$ corresponds to a different funtion $f(x)$, and hence would result in a different dataset. The prediction error for these datasets are the average absolute error with $f(x)$.
    \item \textbf{Engineered datasets}: In Figure~\ref{fig:QuantumError}, we consider datasets that are engineered to saturate the potential of a quantum ML model. Given the choice of classical kernel $K^{\mathrm{C}}$ that has the smallest geometric difference $g$ with a quantum ML model $K^{\mathrm{Q}}$, we can create a data set that saturates $s_{\mathrm{C}} = g^2 s_{\mathrm{Q}}$ following the procedure in Appendix~\ref{app:engdata}. In particular, we construct the dataset such that $s_{\mathrm{Q}}=1$ and $s_{\mathrm{C}} = g^2$. We compute the eigenvector $v$ corresponding to the maximum eigenvalue of
    \begin{equation}
        \sqrt{K^{\mathrm{Q}}} \sqrt{K^{\mathrm{C}}} \left(K^{\mathrm{C}} + \lambda I\right)^{-2} \sqrt{K^{\mathrm{C}}} \sqrt{K^{\mathrm{Q}}}
    \end{equation}
    and construct $y' = \sqrt{K^{\mathrm{Q}}} v \in \mathbb{R}^N$. $y'_i$ corresponds to a real number for data point $x_i$.
    Finally we define the label of input data point $x_i$ as
    \begin{equation}
        y_i = \begin{cases} \mathrm{sign}(y'_i), & \text{with prob. } 0.9, \\ \text{random} \pm 1, & \text{with prob. } 0.1. \end{cases}
    \end{equation}
    This data set will show the maximal separation between quantum and classical ML model.
    The plots in Figure~\ref{fig:QuantumError} uses engineered datasets generated by saturating the geometric difference of classical ML models and quantum ML models based on projected quantum kernels in Equation~\eqref{eq:PQKnumerics} under different embeddings (E1, E2, and E3).
    In Figure~\ref{fig:QuantumError-QK}, we show the results for quantum ML models based on the original quantum kernels.
\end{enumerate}

\begin{figure*}[t]
\centering
\includegraphics[width=0.96\textwidth]{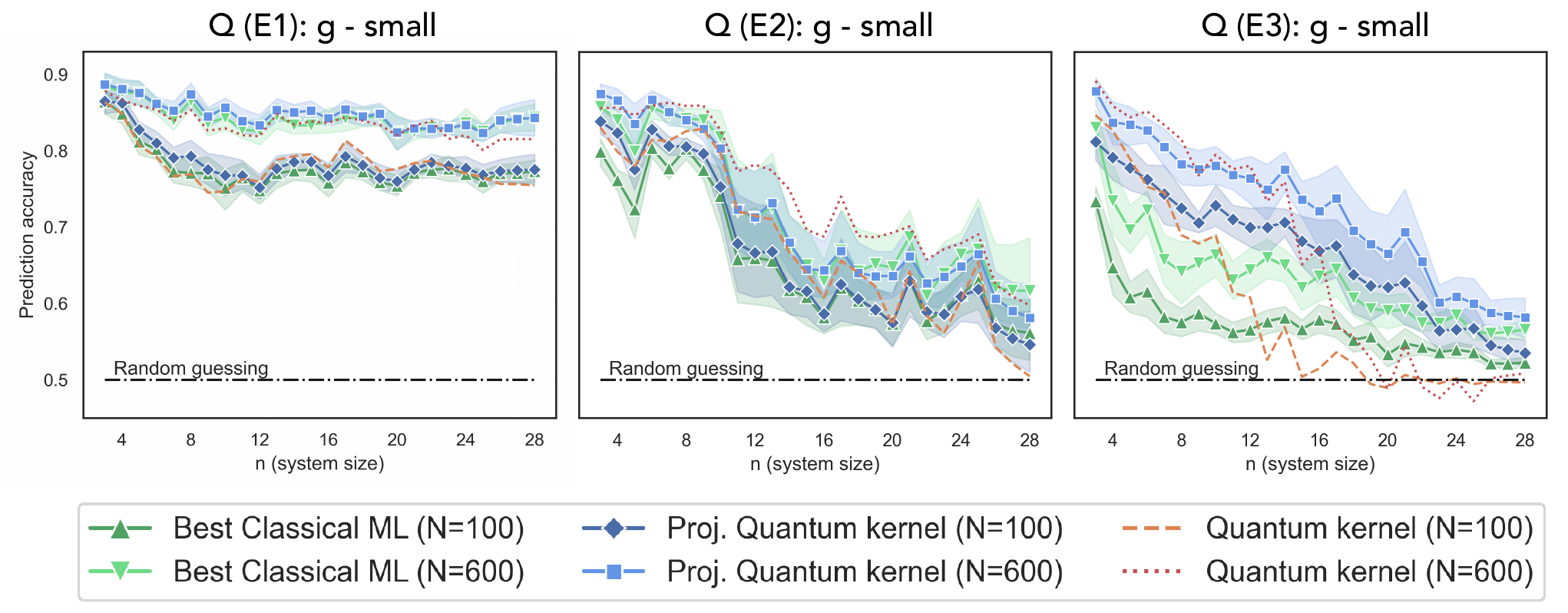}
    \caption{Prediction accuracy (higher the better) on engineered data sets. A label function is engineered to match the geometric difference $g(\mathrm{C} || \mathrm{QK})$ between the original quantum kernel and classical approaches. No substantial advantage is found using quantum kernel methods at large system size due to the small geometric difference $g(\mathrm{C} || \mathrm{QK})$. We consider the best performing classical ML models among Gaussian SVM, linear SVM, Adaboost, random forest, neural networks, and gradient boosting. \label{fig:QuantumError-QK}}
\end{figure*}

\subsection{Classical machine learning models}
\label{sec:CMLmodel}

We present the list of classical machine learning models that we compared with. We used scikit-learn \cite{sklearn_api} for training the classical ML models.
\begin{itemize}
    \item Neural network: We perform a grid search over two-layer feedforward neural networks with hidden layer size
    \begin{equation}
        h \in \{ 10, 25, 50, 75, 100, 125, 150, 200 \}.
    \end{equation}
    For classification, we use $\mathrm{MLPClassifier}$. For regression, we use $\mathrm{MLPRegressor}$.
    \item Linear kernel method: We perform a grid search over the regularization parameter
    \begin{equation}
        C \in \{0.006, 0.015, 0.03, 0.0625, 0.125, 0.25, 0.5, 1.0, 2.0, 4.0, 8.0, 16.0, 32.0, 64.0, 128.0, 256, 512, 1024\}.
    \end{equation}
    For classification, we use $\mathrm{SVC}$ with linear kernel. For regression, we choose the best between SVR and KernelRidge (both using linear kernel).
    \item Gaussian kernel method: We perform a grid search over the regularization parameter
    \begin{equation}
        C \in \{0.006, 0.015, 0.03, 0.0625, 0.125, 0.25, 0.5, 1.0, 2.0, 4.0, 8.0, 16.0, 32.0, 64.0, 128.0, 256, 512, 1024\}.
    \end{equation}
    and kernel hyper-parameter
    \begin{equation}
        \gamma \in \{0.25, 0.5, 1.0, 2.0, 3.0, 4.0, 5.0, 20.0\}  / (n \Var[x_{ik}]).
    \end{equation}
    $\Var[x_{ik}]$ is the variance of all the coordinates $k=1, \ldots, n$ from all the data points $x_1, \ldots, x_N$.
    For classification, we use $\mathrm{SVC}$ with RBF kernel (equivalent to Gaussian kernel). For regression, we choose the best between SVR and KernelRidge (both using RBF kernel).
    \item Random forest: We perform a grid search over the individual tree depth
    \begin{equation}
        \text{max\_depth} \in \{2, 3, 4, 5\},
    \end{equation}
    and number of trees
    \begin{equation}
        \text{n\_estimators} \in \{25, 50, 100, 200, 500\}.
    \end{equation}
    For classification, we use $\mathrm{RandomForestClassifier}$. For regression, we use $\mathrm{RandomForestRegressor}$.
    \item Gradient boosting: We perform a grid search over the individual tree depth
    \begin{equation}
        \text{max\_depth} \in \{2, 3, 4, 5\},
    \end{equation}
    and number of trees
    \begin{equation}
        \text{n\_estimators} \in \{25, 50, 100, 200, 500\}.
    \end{equation}
    For classification, we use $\mathrm{GradientBoostingClassifier}$. For regression, we use $\mathrm{GradientBoostingRegressor}$.
    \item Adaboost: We perform a grid search over the number of estimators
    \begin{equation}
        \text{n\_estimators} \in \{25, 50, 100, 200, 500\}.
    \end{equation}
    For classification, we use $\mathrm{AdaBoostClassifier}$. For regression, we use $\mathrm{AdaBoostRegressor}$.
\end{itemize}

\subsection{Quantum machine learning models}

For training quantum kernel methods, we use the kernel function $k^\mathrm{Q}(x_i, x_j) = \Tr(\rho(x_i) \rho(x_j))$. For classification, we use $\mathrm{SVC}$ with the quantum kernel. For regression, we choose the best between SVR and KernelRidge (both using the quantum kernel).
We perform a grid search over 
    \begin{equation}
        C \in \{0.006, 0.015, 0.03, 0.0625, 0.125, 0.25, 0.5, 1.0, 2.0, 4.0, 8.0, 16.0, 32.0, 64.0, 128.0, 256, 512, 1024\}.
    \end{equation}
For training projected quantum kernel methods, we use the kernel function
\begin{align}
    k^{\text{PQ}}(x_i, x_j) = \exp \left(-\gamma  \sum_k \sum_{P \in \{X, Y, Z\}} \left(\Tr(P \rho(x_i)_k) - \Tr( P \rho(x_j)_k)\right)^2 \right),
\end{align}
where $P$ is a Pauli matrix.
For classification, we use $\mathrm{SVC}$ with the projected quantum kernel $k^{\text{PQ}}(x_i, x_j)$. For regression, we choose the best between SVR and KernelRidge (both using the projected quantum kernel $k^{\text{PQ}}(x_i, x_j)$).
We perform a grid search over 
\begin{equation}
        C \in \{0.006, 0.015, 0.03, 0.0625, 0.125, 0.25, 0.5, 1.0, 2.0, 4.0, 8.0, 16.0, 32.0, 64.0, 128.0, 256, 512, 1024\}.
\end{equation}
and kernel hyper-parameter
\begin{equation}
    \gamma \in \{0.25, 0.5, 1.0, 2.0, 3.0, 4.0, 5.0, 20.0\}  / (n \Var[\Tr(P \rho(x_i)_k)]).
\end{equation}
$\Var[\Tr(P \rho(x_i)_k)]$ is the variance of $\Tr(P \rho(x_i)_k)$ for all $P \in \{X, Y, Z\}$, all coordinates $k=1, \ldots, n$, and all data points $x_1, \ldots, x_N$.
We report the prediction performance under the best hyper-parameter for all classical and quantum machine learning models.

\section{Additional numerical experiments}
\label{app:additional}

\begin{figure*}[t]
\centering
\includegraphics[width=0.96\textwidth]{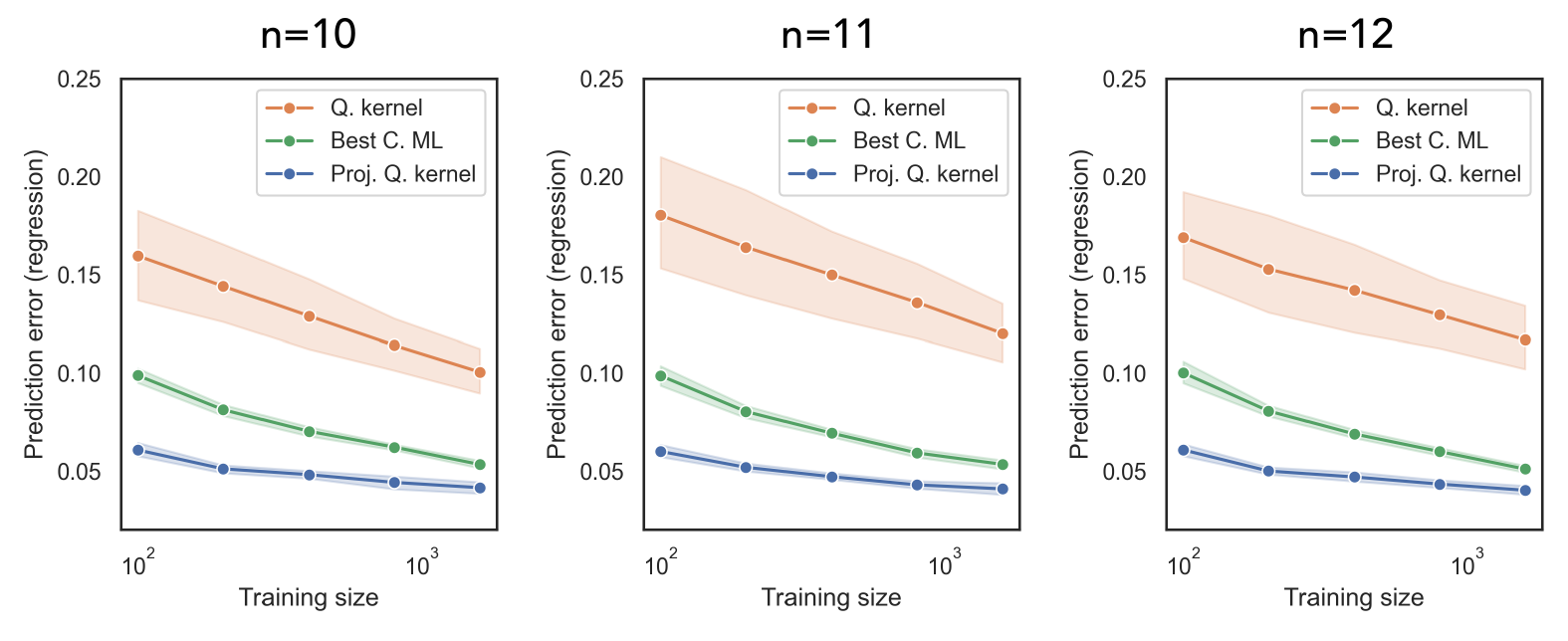}
    \caption{Prediction error (lower the better) on quantum data set (E2) over different training set size $N$. We can see that as the number of data increases, every model improves and the separation between them decreases. \label{fig:power-of-data}}
\end{figure*}

\begin{figure*}[t]
\centering
\includegraphics[width=0.96\textwidth]{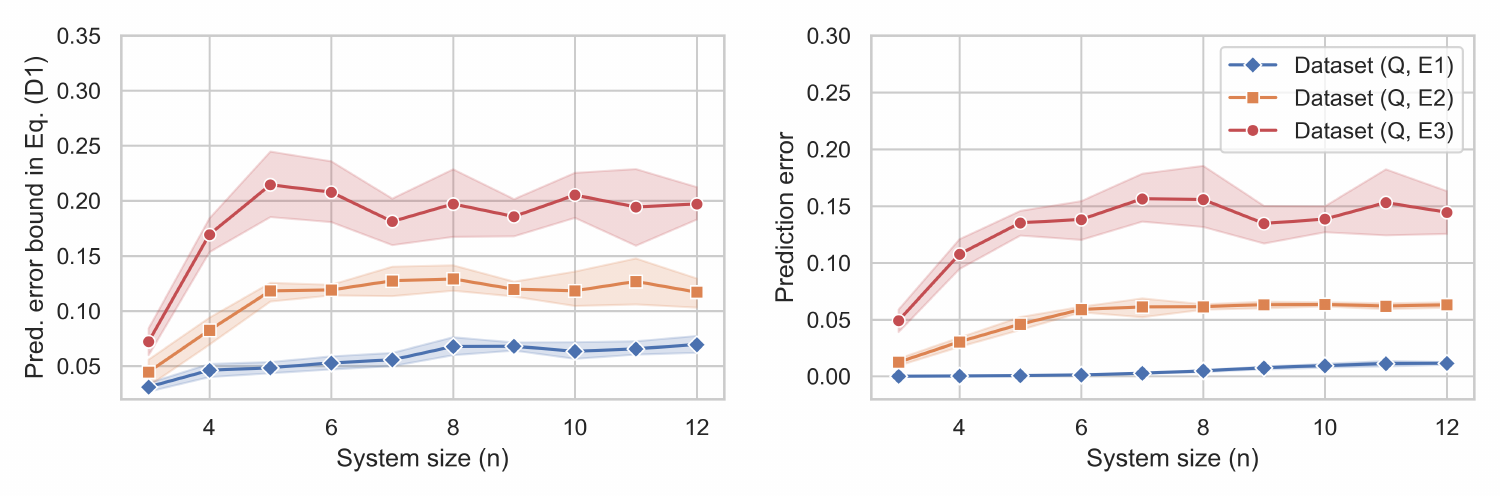}
    \caption{A comparison between the prediction error bound based on classical kernel methods (see Eq.~\eqref{eq:mainprederr}) and the prediction performance of the best classical ML model on the three quantum datasets. We consider the best performing classical ML models among Gaussian SVM, linear SVM, Adaboost, random forest, neural networks, and gradient boosting. While the prediction error bound is an upper bound to the actual prediction error, the trends are very similar (a large prediction error bound gives a large prediction error). \label{fig:S-vs-Pred}}
\end{figure*}

In the main text, we have presented engineered data sets to saturate the geometric inequality $s_{\mathrm{C}} \leq g(\mathrm{C} || \mathrm{PQ})^2 s_{\mathrm{PQ}}$ between classical ML and projected quantum kernel.
As an additional experiment to see if the same approach can work with the original quantum kernel method, we can create similar engineered data sets that saturate the geometric inequality between classical ML and quantum kernel
The result is given in Fig.~\ref{fig:QuantumError-QK}.
We can see that due to the large dimension $d$ and small geometric difference $g(\mathrm{C} || \mathrm{Q})$ between classical ML and quantum kernel at large system size, there are no obvious advantage even for this best-case scenario.
Interestingly, we see some advantage of projected quantum kernel over classical ML even when this data set is not constructed for projected quantum kernel.

In Fig.~\ref{fig:power-of-data}, we show the prediction performance for learning a quantum neural network under a wide range for the number of training data $N$.
We can see that there is a non-trivial advantage for small training size $N=100$ when comparing projected quantum kernel and the best classical ML model.
However, as training size $N$ increases, every model will improve and the prediction advantage will shrink.

In Fig.~\ref{fig:S-vs-Pred}, we compare the prediction error bound $s_K(N)$ for classical kernel methods and the prediction performance of the best classical ML model (including a variety of classical ML models in Section~\ref{sec:CMLmodel}).
To be more precise, we consider different classical kernel functions and different regularization parameter $\lambda$.
Then we compute
\begin{equation}
    s_{K, \lambda}(N) = \sqrt{ \frac{\lambda^2 \sum_{i=1}^N \sum_{j=1}^N ((K + \lambda I)^{-2})_{ij} y_i y_j}{N}} + \sqrt{\frac{\sum_{i=1}^N \sum_{j=1}^N ((K + \lambda I)^{-1} K (K + \lambda I)^{-1})_{ij} y_i y_j}{N}}.
\end{equation}
This is a generalization of $s_K(N)$ described in the main text, where we consider regularized classical kernel methods with a regularization parameter $\lambda$ to improve generalization performance (setting $\lambda = 0$ reduces to $s_K(N)$ given in the main text). See Section~\ref{sec:proofmain} for a detailed proof of an upper bound to the prediction error (note that the output label $y_i = \Tr(O^U \rho(x_i))$).
We can see that while the prediction error bound and the actual prediction error has a non-negligible gap, the two figures follow a similar trend.
When the prediction error bound is small, the prediction error of the best classical ML is also fairly small (and vice versa). It shows that $s_{K, \lambda}(N)$ is a good predictive metric for whether a classical ML model can learn to predict outputs from a quantum computation model.

\end{document}